\documentclass[11pt,reqno]{amsart}
\usepackage{amsmath,amsthm,amssymb,bm,bbm,dsfont,braket,esint,nicefrac}
\usepackage[mathscr]{eucal}
\usepackage{amsaddr}
\usepackage{upgreek}
\usepackage[left=2cm,right=2cm,top=2cm,bottom=2cm]{geometry}
\usepackage{mathtools}\mathtoolsset{centercolon}
\usepackage[nospace,noadjust]{cite}
\usepackage{mleftright}
\mleftright 

\usepackage{pifont} 

\usepackage[shortlabels]{enumitem}
\usepackage{setspace}
\usepackage{graphicx}
\usepackage{verbatim}
\usepackage{float}
\usepackage{placeins}
\usepackage{array}
\usepackage{booktabs}
\usepackage{multicol,multirow,tabularx}
\usepackage{threeparttable}
\usepackage[update,prepend]{epstopdf}
\usepackage{amsfonts,xfrac,comment}
\usepackage[abs]{overpic}
\usepackage{xspace}

\usepackage{tcolorbox}
\newtcolorbox{mytcolorbox}{
top = 1pt,
right = 1pt,
left = 1pt,
bottom = 1pt,
colback = white, 
boxrule = 0.5pt,
}

\makeatletter
\def\adl@drawiv#1#2#3{%
	\hskip0
	\tabcolsep
	\xleaders#3{#2 0\@tempdimb #1{1}#2 0.5\@tempdimb}%
	#2\z@ plus1fil minus1fil\relax
	\hskip0\tabcolsep}
\newcommand{\cdashlinelr}[1]{%
	\noalign{\vskip\aboverulesep
		\global\let\@dashdrawstore\adl@draw
		\global\let\adl@draw\adl@drawiv}
	\cdashline{#1}
	\noalign{\global\let\adl@draw\@dashdrawstore
		\vskip\belowrulesep}}
\makeatother

\usepackage[dvipsnames]{xcolor}
\usepackage{colortbl}
\usepackage{arydshln}
\usepackage[hidelinks,linktocpage=true]{hyperref}
\hypersetup{
	unicode=false,          
	pdftoolbar=true,        
	pdfmenubar=true,        
	pdffitwindow=false,     
	pdfstartview={FitH},    
	pdfnewwindow=true,      
	colorlinks=true,        
	linkcolor=Red,          
	citecolor=ForestGreen,  
	filecolor=Magenta,      
	urlcolor=BlueViolet,    
}
\usepackage{cleveref}

\usepackage{doi}
\usepackage{url}
\usepackage{caption, subcaption}
\usepackage{shadethm}
\usepackage{pgf}

\makeatletter
\ifx\@NODS\undefined%

\let\mathbb=\mathds
\else%
\fi
\makeatother

\def\d{{\text {\rm d}}}

\newcommand{\NS}{{Nussbaum--Szko{\l}a}\xspace}

\DeclareMathOperator*{\argmax}{arg\,max}
\DeclareMathOperator*{\argmin}{arg\,min} 

\DeclareMathOperator{\Err}{Err}
\DeclareMathOperator{\HM}{HM}
\newcommand{\Succ}{\textsc{Succ}}
\newcommand{\Fail}{\textsc{Fail}}
\newcommand{\Loss}{\textsc{Loss}}
\newcommand{\Min}{\textsc{Min}}

\DeclareMathOperator{\Tr}{Tr}

\DeclareMathOperator{\I}{\mathbf{1}}
\DeclareMathOperator{\supp}{supp}
\DeclareMathOperator{\e}{\mathrm{e}}

\DeclareMathOperator{\Var}{Var}

        \def\cH{{\cal H}}

\def\0{{\mathbf{0}}}
\def\1{{\mathbf{1}}}
\def\2{{\mathbf{2}}}
\def\3{{\mathbf{3}}}
\def\4{{\mathbf{4}}}
\def\5{{\mathbf{5}}}
\def\6{{\mathbf{6}}}

\def\7{{\mathbf{7}}}
\def\8{{\mathbf{8}}}
\def\9{{\mathbf{9}}}

\def\bea{\begin{eqnarray}}
\def\eea{\end{eqnarray}}

\def\eps{\varepsilon}

\def\cH{{\mathcal H}}

\theoremstyle{plain}
\newshadetheorem{theorem}{Theorem}[section]
\newshadetheorem{proposition}[theorem]{Proposition}
\newshadetheorem{lemma}[theorem]{Lemma}
\newtheorem{fact}[theorem]{Fact} 
\newshadetheorem{corollary}[theorem]{Corollary}

\theoremstyle{definition}
\newtheorem{definition}[theorem]{Definition} 

\theoremstyle{remark}
\newtheorem{remark}[theorem]{Remark}

\makeatletter
\newcommand{\opnorm}{\@ifstar\@opnorms\@opnorm}
\newcommand{\@opnorms}[1]{%
	$\left|\mkern-1.5mu\left|\mkern-1.5mu\left|
	#1
	\right|\mkern-1.5mu\right|\mkern-1.5mu\right|$
}
\newcommand{\@opnorm}[2][]{%
	\mathopen{#1|\mkern-1.5mu#1|\mkern-1.5mu#1|}
	#2
	\mathclose{#1|\mkern-1.5mu#1|\mkern-1.5mu#1|}
}
\makeatother

\usepackage{tikz}
\usetikzlibrary{arrows.meta} 
\tikzset{>={Latex[length=4,width=4]}} 
\usetikzlibrary{shapes.geometric, positioning, calc} 
\usetikzlibrary{arrows.meta, calc, positioning, backgrounds, fit}

\definecolor{statebg}{RGB}{252,231,231}   
\definecolor{measbg}{RGB}{238,233,247}     
\definecolor{statetitle}{RGB}{223,79,38}   
\definecolor{meastitle}{RGB}{118,49,148}   
\definecolor{boxborder}{RGB}{38,58,110}    
\definecolor{fangray}{RGB}{150,150,150}    

\definecolor{ink}{RGB}{33,37,43}          
\definecolor{stateacc}{RGB}{43,76,126}    
\definecolor{measacc}{RGB}{118,52,108}    
\definecolor{statefill}{RGB}{244,247,251} 
\definecolor{measfill}{RGB}{249,246,251}  
\definecolor{paneledge}{RGB}{208,214,224} 
\definecolor{arrowgray}{RGB}{120,126,138} 
\definecolor{labelgray}{RGB}{96,101,112}  

\newcommand{\sh}{\kern-0.08em$^\textbf{\#}$\hspace{-3pt}}
\renewcommand{\b}{\kern-0.06em$\flat$}

\begin{document}

\let\origmaketitle\maketitle
\def\maketitle{
	\begingroup
	\def\uppercasenonmath##1{} 
	\let\MakeUppercase\relax 
	\origmaketitle
	\endgroup
}

\title{\bfseries \Large{ 
Multiple Quantum Hypothesis Testing:\\ One-Shot Pairwise Bounds and Sharp Asymptotics
		}}

\author{ \normalsize 
{Hao-Chung Cheng}$^{1\text{--}5}$
and
{Po-Chieh Liu}$^{1,2}$
}
\address{\small  	
$^1$Department of Electrical Engineering and Graduate Institute of Communication Engineering,\\ National Taiwan University, Taipei 106, Taiwan (R.O.C.)\\
$^2$Department of Mathematics, National Taiwan University\\
$^3$Center for Quantum Science and Engineering, National Taiwan University\\
$^4$Hon Hai (Foxconn) Quantum Computing Center, New Taipei City 236, Taiwan (R.O.C.)\\
$^5$Physics and Mathematics Divisions, National Center for Theoretical Sciences, Taipei 10617, Taiwan (R.O.C.)\\
}

\email{\href{mailto:haochung.ch@gmail.com}{haochung.ch@gmail.com}}

\date{\today}

\begin{abstract}
We consider Bayesian discrimination among multiple quantum states and establish a dimension-free one-shot upper bound on the minimum probability of error in terms of the sum of pairwise errors.
This resolves a conjecture of Audenaert and Mosonyi \href{https://doi.org/10.1063/1.4898559}{[J.~Math.~Phys.~55 (2014)]} and improves the multiple quantum Chernoff bound of Li \href{https://doi.org/10.1214/16-AOS1436}{[Ann.~Statist.~44 (2016)]} by removing its dimension-dependent prefactor.
In the asymptotic many-copy regime, our bound proves the achievability of the multiple quantum Chernoff distance for arbitrary separable Hilbert spaces, thereby settling the previously open infinite-dimensional case, and further yields constant-factor sharp asymptotics for the optimal error probability.

In binary quantum hypothesis testing, we prove that the minimum error probability is characterized, up to universal constants, by a trace harmonic-mean quantity. 
Consequently, the optimal binary quantum error probability is within a factor of two of the optimal classical error probability for the associated Nussbaum--Szko{\l}a distributions, complementing the lower bound of Nussbaum and Szko{\l}a \href{https://doi.org/10.1214/08-AOS593}{[Ann.~Statist.~37 (2009)]}.
%

\end{abstract}

\maketitle
\tableofcontents

\newpage
\section{Introduction} \label{sec:introduction}

Discriminating among multiple quantum hypotheses is a cornerstone of quantum information science \cite{Hel67, Hol72, Hol74, Hol76, YKL75, HP91, ON00, Parthasarathy_2001, Nussbaum2013attainment, Montanaro2019, NS11, NussbaumSzkola2010spin_chain, NussbaumSzkola2011pure, AM14, Li16, Cheng_simple, PerezGuijarro2022, PerezGuijarro2025, sample_complexity_25}.
Consider a quantum system prepared under one of $M$ distinct hypotheses, $\{\mathtt{H}_i\}_{i=1}^M$. 
Under hypothesis $\mathtt{H}_i$, which occurs with prior probability $p_i \in (0,1)$, the system is described by a state $\rho_i$, where each $\rho_i$ is a density operator on a separable, possibly infinite-dimensional Hilbert space.
To distinguish between these hypotheses, the statistician's objective is to design a quantum measurement, described by a positive operator-valued measure (POVM) $\{\Lambda_i\}_{i=1}^M$ satisfying $\Lambda_i\geq 0$ and $\sum_{i=1}^M \Lambda_i = \I$; see Figure~\ref{figure:QHT}.
The performance of $M$-ary quantum hypothesis testing is quantified by the minimum probability of decision error, optimized over all possible quantum measurements:
\begin{align} \label{eq:minimum_error}
\Err^\star \left(p_1 \rho_1,\dots, p_M \rho_M\right)
\coloneqq \inf_{\mathrm{POVM}\, \{\Lambda_i\}_{i=1}^M} 1 - \sum_{i=1}^M p_i \Tr[\rho_i \Lambda_i].
\end{align}

Considerable research has been devoted to deriving error bounds for \eqref{eq:minimum_error}.
For binary hypothesis testing (i.e., $M=2$) over finite-dimensional quantum systems, the celebrated \emph{quantum Chernoff bound} establishes that \cite{Che52}, \cite[Theorem 1]{ACM+07}, \cite[Theorem 2]{ANS+08}:
\begin{align} \label{eq:Chernoff0}
\Err^\star \left(p_1 \rho_1, p_2 \rho_2\right)
&\leq \e^{- C(\rho_1, \rho_2)},
\end{align}
where the quantum Chernoff distance is defined as
\begin{align} \label{eq:Chernoff_distance}
C(\rho_1, \rho_2) 
&\coloneqq \sup_{\alpha \in (0,1)} - \log \Tr\left[ \rho_1^{\alpha} \rho_2^{1-\alpha}\right].
\end{align}
This bound was subsequently generalized to infinite-dimensional systems \cite[Proposition 1.1]{JOP+12}, \cite[p.~270]{JOP+12_book}, \cite[Theorem 4.54]{HP14}.
When the quantum systems are prepared in independent and identically distributed (i.i.d.) states $\rho_i^{\otimes n}$, the additivity of the quantum Chernoff distance implies that $\Err^\star \left(p_1 \rho_1^{\otimes n}, p_2 \rho_2^{\otimes n}\right)$ decays exponentially at the rate $C(\rho_1, \rho_2)$. 
Nussbaum and Szko{\l}a further proved that this decay rate is asymptotically optimal \cite{NS09}.
The fundamental significance of the quantum Chernoff bound \eqref{eq:Chernoff0} lies in its establishment of the exponentially tight \emph{large-deviation principle} in quantum information science.
This landmark result has since become a cornerstone of error exponent analysis, driving theoretical advances for numerous research topics, including sample and query complexities \cite{sample_complexity_25, TheshaniWilde2025Query}, composite hypothesis testing \cite{HT14, Fang2025generalized, FangHayashi2026Correlated}, quantum channel discrimination \cite{Wilde2020Amortized, TheshaniWilde2025Query}, quantum metrology \cite{Meyer2025Metrology}, communication over quantum channels \cite{Hay07, QWW18, Cheng_simple, preparation}, moderate-deviation analysis \cite{CH17, CTT2017}, entanglement distillation \cite{lin2026entanglement}, and strong converses in covering-type problems \cite{SGC22a, CG22, CG23, CDG24}.

The error analysis for multiple quantum hypothesis testing ($M > 2$), however, is considerably more challenging due to the absence of closed-form expressions for the optimal quantum measurements \cite{YKL75}.
Nussbaum and Szko{\l}a established that the optimal error exponent for multiple hypothesis testing cannot exceed that of the worst-case pair within the ensemble \cite{NussbaumSzkola2011pure}.
That is, the multiple quantum Chernoff distance, defined as
\begin{align} \label{eq:Chernoff_multiple}
C(\rho_1, \ldots, \rho_M)
\coloneqq \min_{(i,j):\, i\neq j} C(\rho_i, \rho_j),
\end{align}
constitutes the fundamental upper bound on the asymptotic decay rate, which they conjectured to be achievable by a quantum measurement.
The conjecture was initially proved for several special cases, including pure-state ensembles and states with mutually orthogonal supports \cite{NS11, NussbaumSzkola2011pure}.
For general quantum ensembles, Nussbaum and Szko{\l}a first demonstrated that a rate of $\frac{1}{3} C(\rho_1, \ldots, \rho_M)$ is generally achievable \cite{NS11}, a bound subsequently improved to $\frac{1}{2} C(\rho_1, \ldots, \rho_M)$ by Audenaert and Mosonyi \cite{AM14}.
Groundbreaking progress was made by Ke Li, who proved that the multiple quantum Chernoff distance is achievable for any finite-dimensional quantum ensemble \cite{Li16}.
Essentially, Li established a one-shot upper bound on \eqref{eq:minimum_error} in terms of pairwise errors, multiplied by a dimension-dependent prefactor.
In the i.i.d.~scenario, this factor grows polynomially with the number of copies $n$, thereby leaving the asymptotic decay rate unaffected.
Nevertheless, for systems operating in infinite dimensions, the dimension-dependent nature of this technique renders the bound inapplicable.
Standard finite-rank approximation does not straightforwardly  remedy the multiple Chernoff problem.
Consequently, whether the multiple quantum Chernoff bound holds in its most general form has remained an open problem prior to the present work. 
This unresolved question leaves a critical gap in the foundational theory of various quantum state discrimination and communication tasks \cite{Zhuang2021Quantum, Chang2025Joint}.
To bridge this gap, the central objective of the present work is to address the following question:
\begin{align} \label{eq:Question} \tag{Q}
\textit{Is the multiple Chernoff distance in \eqref{eq:Chernoff_multiple} achievable in infinite-dimensional quantum systems?}
\end{align}

In this paper, we answer \eqref{eq:Question} in the affirmative.
Our first main result is to establish a general one-shot error bound for infinite dimensions, with to only a universal multiplicative constant prefactor
(Theorem~\ref{theorem:pairwise}):
\begin{mytcolorbox}
\begin{align} \label{eq:pairwise0}
\Err^{\star}\left(p_1 \rho_1, \ldots, p_M \rho_M\right)
\leq 4 \sum_{(i,j):\,i<j} \Err^\star \!\big(p_i P^{ij}, p_j Q^{ij}\big)
\leq
\begin{dcases}
8\sum_{(i,j):\, i<j}
\Err^\star \left(p_i \rho_i, p_j \rho_j\right),
\\
4\sum_{(i,j):\, i<j} \e^{- C(\rho_i, \rho_j)}.
\end{dcases}
\end{align}
\end{mytcolorbox}
\noindent Here, $(P^{ij}, Q^{ij})$ denote the \NS distributions for the pair of states $(\rho_i, \rho_j)$ \cite{NS09}, detailed in Definition~\ref{defn:NS} below (see also the recent development \cite{AJ23, Anastasiadis2026}).
The upper branch of \eqref{eq:pairwise0} resolves a \emph{pairwise bound} conjecture proposed by Audenaert and Mosonyi \cite[Conjecture 2.3]{AM14}, and the lower branch directly yields the achievability of the multiple Chernoff distance $C(\rho_1, \ldots, \rho_M)$.
The key ingredients of the proof are the design of a new quantum measurement and a judicious application of the \emph{quantum union bound} \cite{Gao15, OMW19, OV22}.

\medskip
Our second contribution revisits both the quantum Chernoff bound \cite{ACM+07, ANS+08} and Nussbaum and Szko{\l}a's lower bound \cite{NS09} in the binary setting ($M=2$).
For notational simplicity, letting $A = p_1 \rho_1$ and $B = p_2 \rho_2$, we demonstrate that the minimum error $\Err^{\star}(A,B)$ is tightly bounded from both above and below by the \emph{Petz--\NS harmonic mean}:
\begin{align}
    \HM_{s}(A,B)
    \coloneqq \Tr\left[ \left( (1-s) P^{-1} + s Q^{-1} \right)^{-1} \right],
\end{align}
where $(P,Q)$ are the \NS distributions corresponding to $(A,B)$.
More precisely, we establish the following chain of inequalities (Theorem~\ref{theorem:harmonic-mean_bound}):
\begin{mytcolorbox}
\begin{align} \label{eq:harmonic-mean_bound0}
    s\wedge (1-s) \cdot \Err^\star \left(P, Q\right)
		\leq s\wedge (1-s) \cdot \HM_{s}(A,B)
		\leq \Err^\star \left(A, B\right)
        \leq \HM_{s}(A,B) 
        \leq
        \begin{dcases}
        \frac{\Err^\star \left(P, Q\right)}{s\wedge (1-s)} ,
        \\
        \Tr\left[A^{1-s} B^s\right].
        \end{dcases}
\end{align}
\end{mytcolorbox}
These bounds bridge several existing results in the literature.
When specializing to $s = \frac12$, the leftmost of \eqref{eq:harmonic-mean_bound0} recovers Nussbaum--Szko{\l}a's lower bound $\frac12 \Err^{\star}(P,Q) \leq \Err^{\star}(A,B)$.
On the other hand, the rightmost lower part of \eqref{eq:harmonic-mean_bound0} directly implies the quantum Chernoff bound \eqref{eq:Chernoff0} via the scalar harmonic--geometric mean inequality and optimizing over $s\in(0,1)$.
Crucially, evaluating the rightmost upper part of \eqref{eq:harmonic-mean_bound0} at $s = \frac12$ complements Nussbaum--Szko{\l}a's lower bound by providing a tight upper bound:
\begin{align} \label{eq:control}
    \frac12 \cdot \Err^{\star}(P,Q) \leq \Err^{\star}(A,B)
    \leq 2 \cdot \Err^{\star}(P,Q).
\end{align}
This demonstrates a profound equivalence at the one-shot level: the minimum error of binary quantum hypothesis testing is guaranteed to be within a factor of $2$ of its classical counterpart operating over the \NS distributions.

The one-shot bounds derived in \eqref{eq:control} enable us to provide a more accurate estimate of the minimum error in the i.i.d.~setting than the quantum Chernoff bound \eqref{eq:Chernoff0}.
Under the standard Cram\'er's conditions (see e.g.~\cite{BR60}, \cite{Petrov1965}, \cite[Chapters VII]{Petrov1975}, \cite[\S 2.2]{DZ98}), \eqref{eq:control} implies, for the first time, the almost-exact i.i.d.~asymptotics of general binary quantum hypothesis testing (Theorem~\ref{theorem:asymptotics_symmetric}):
\begin{mytcolorbox}
    \begin{align}
        \Err^\star\left(p_1 \rho_1^{\otimes n}, p_2 \rho_2^{\otimes n} \right)
        \asymp \frac{p_1^{\alpha^{\star}}p_2^{1-\alpha^\star}}{\alpha^{\star}(1-\alpha^{\star})\sqrt{2\pi n V_{\alpha^\star}(\rho_1\Vert\rho_2)}} \e^{ -n C(\rho_1, \rho_2) }.
    \end{align}
\end{mytcolorbox}
\noindent Here, $f(n)\asymp g(n)$ indicates that there exists a constant $c>0$ independent of $n$ satisfying  $\frac{1}{c} g(n) \leq f(n) \leq c g(n)$ for all sufficiently large $n \in \mathbb{N}$. 
The parameter $\alpha^{\star} \in (0,1)$ is the unique optimizer for the quantum Chernoff distance in \eqref{eq:Chernoff_distance}, and $V_{\alpha^\star}(\rho_1\Vert\rho_2)$ represents the associated second-moment quantity.
Applying a similar line of reasoning to multiple quantum hypothesis testing over infinite-dimensional systems, and combining \eqref{eq:pairwise0} with the Nussbaum--Szko{\l}a lower bound \cite{NS09, NussbaumSzkola2011pure}, we establish the following almost-exact asymptotic characterization:
\begin{mytcolorbox}
    \begin{align}
        \Err^\star\left(p_1 \rho_1^{\otimes n}, \cdots, p_M \rho_M^{\otimes n} \right)
        \asymp \sum_{(i,j):\, C(\rho_i,\rho_j) = C(\rho_1, \ldots, \rho_M)}\frac{p_i^{\alpha_{ij}^\star}p_j^{1-\alpha_{ij}^\star}}{\alpha_{ij}^{\star}(1-\alpha_{ij}^{\star})\sqrt{2\pi n V_{\alpha_{ij}^\star}(\rho_i\Vert\rho_j)}} \e^{ -n C(\rho_1, \ldots, \rho_M) }.
    \end{align}
\end{mytcolorbox}
\noindent Here, $\alpha_{ij}^{\star} \in (0,1)$ denotes the unique optimizer for the pairwise quantum Chernoff distance $C(\rho_i,\rho_j)$ in~\eqref{eq:Chernoff_distance}.

\medskip
The remainder of the paper is organized as follows.
Section~\ref{sec:notation} introduces necessary notation and the \NS distributions.
Section~\ref{sec:multiple} establishes the pairwise bounds for multiple quantum hypothesis testing.
Section~\ref{sec:simple} proves the harmonic-mean upper bound for binary quantum hypothesis testing.
Section~\ref{sec:asymptotics} presents the almost-exact asymptotics for both binary and multiple hypothesis testing.
Section~\ref{sec:asymmetric} investigates asymmetric hypothesis testing.
Several technical components and supplementary proofs are deferred to the appendices. Specifically, Appendix~\ref{sec:union_bound} details the quantum union bound alongside an alternative proof based on the data processing inequality. Appendix~\ref{sec:NS} revisits the Nussbaum--Szko{\l}a lower bound. Appendix~\ref{sec:representation} derives various representations of the Petz--\NS harmonic mean that are crucial to the main results in Section~\ref{sec:simple}. Finally, Appendix~\ref{sec:Chernoff} provides an independent derivation of the quantum Chernoff bound.

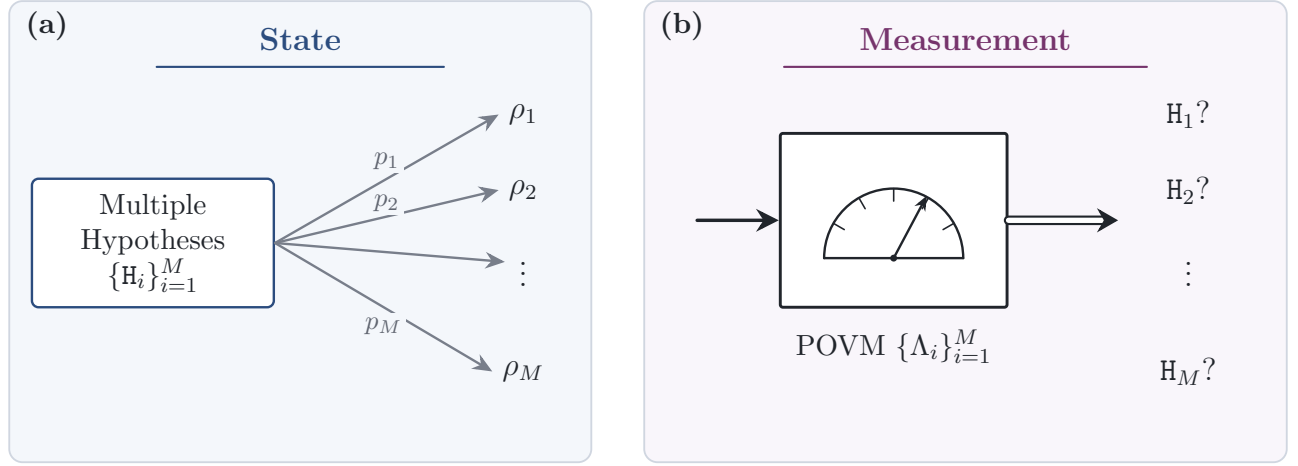
\begin{figure}
\centering
\begin{tikzpicture}[
    >=Stealth,
    line cap=round, line join=round,
    font=\normalsize,
    panel/.style    = {rounded corners=5pt, draw=paneledge, line width=0.7pt},
    ptitle/.style   = {font=\scshape\bfseries\large},
    tag/.style      = {font=\bfseries, text=ink},
    hypbox/.style   = {draw=stateacc, line width=0.9pt, rounded corners=2.5pt,
                       fill=white, align=center, text=ink,
                       minimum width=3.2cm, minimum height=1.7cm},
    fanarrow/.style = {-{Stealth[length=2.6mm,width=2.2mm]},
                       draw=arrowgray, line width=0.9pt},
    plabel/.style   = {font=\small\itshape, text=labelgray,
                       fill=statefill, inner sep=1.5pt},
    measbox/.style  = {draw=ink, line width=1.1pt, rounded corners=1pt,
                       fill=white, minimum width=3.0cm, minimum height=2.3cm},
    inarrow/.style  = {-{Stealth[length=3mm,width=3mm]}, draw=ink, line width=1.3pt},
    outarrow/.style = {-{Stealth[length=2.8mm,width=3.2mm]}, draw=ink,
                       line width=0.8pt, double, double distance=2pt},
    gauge/.style    = {draw=ink, line width=0.9pt},
    state/.style    = {text=ink, font=\large},
]

\fill[statefill,rounded corners=5pt] (-0.35,-3.05) rectangle (7.35,3.05);
\draw[panel] (-0.35,-3.05) rectangle (7.35,3.05);
\node[ptitle,text=stateacc] (stitle) at (3.5,2.55) {State};
\draw[stateacc,line width=0.8pt] (1.6,2.18) -- (5.4,2.18);  
\node[tag] at (0.15,2.72) {(a)};

\node[hypbox] (hyp) at (1.55,-0.15) {Multiple\\Hypotheses\\$\{\mathtt{H}_i\}_{i=1}^M$};

\node[state] (r1) at (6.45, 1.55) {$\rho_1$};
\node[state] (r2) at (6.45, 0.55) {$\rho_2$};
\node[state] (rd) at (6.45,-0.45) {$\vdots$};
\node[state] (rM) at (6.45,-1.85) {$\rho_M$};

\draw[fanarrow] (hyp.east) -- (r1.west);
\draw[fanarrow] (hyp.east) -- (r2.west);
\draw[fanarrow] (hyp.east) -- ($(rd.west)+(-0.05,0.05)$);
\draw[fanarrow] (hyp.east) -- (rM.west);

\node[plabel] at ($(hyp.east)!0.5!(r1.west)+(0,0.24)$) {$p_1$};
\node[plabel] at ($(hyp.east)!0.5!(r2.west)+(0,0.18)$) {$p_2$};
\node[plabel] at ($(hyp.east)!0.5!(rM.west)+(0,-0.24)$) {$p_M$};

\fill[measfill,rounded corners=5pt] (8.05,-3.05) rectangle (16.55,3.05);
\draw[panel] (8.05,-3.05) rectangle (16.55,3.05);
\node[ptitle,text=measacc] at (12.3,2.55) {Measurement};
\draw[measacc,line width=0.8pt] (9.9,2.18) -- (14.7,2.18);
\node[tag] at (8.55,2.72) {(b)};

\node[measbox] (mb) at (11.35,0.15) {};

\draw[inarrow] (8.75,0.15) -- (mb.west);

\coordinate (gc) at ($(mb.center)+(0,-0.5)$);   
\def\gr{0.92}                                     
\draw[gauge] ($(gc)+(-\gr,0)$) arc (180:0:\gr);   
\draw[gauge] ($(gc)+(-\gr,0)$) -- ($(gc)+(\gr,0)$);
\foreach \a in {0,30,...,180}{                     
  \draw[gauge,line width=0.6pt]
    ($(gc)+(\a:\gr)$) -- ($(gc)+(\a:0.82*\gr)$);}
\draw[gauge,-{Stealth[length=2mm,width=1.6mm]}]    
  (gc) -- ($(gc)+(62:1.02*\gr)$);
\fill[ink] (gc) circle (1.3pt);                    

\node[text=ink] at ($(mb.south)+(0,-0.5)$)
  {POVM $\{\Lambda_i\}_{i=1}^{M}$};

\draw[outarrow] (mb.east) -- ($(mb.east)+(1.45,0)$);

\node[state] (h1) at (15.25, 1.55) {$\mathtt{H}_1?$};
\node[state] (h2) at (15.25, 0.55) {$\mathtt{H}_2?$};
\node[state] (hd) at (15.25,-0.45) {$\vdots$};
\node[state] (hM) at (15.25,-1.85) {$\mathtt{H}_M?$};

\end{tikzpicture}
\caption{Illustration of the multiple quantum hypothesis testing.}
\label{figure:QHT}
\end{figure}

\section{Notation and Preliminaries} \label{sec:notation}
Throughout the paper, we consider a more general picture of positive semidefinite trace-class operators $A_1,A_2,\dots,A_M$, because one can always identify the weighted state $p_i \rho_i$ by $A_i$.
For two trace-class operators $A$ and $B$, we denote the Hilbert--Schmidt inner product by $\langle A, B \rangle \coloneqq \Tr[A^{\dagger} B]$.
Suppose a self-adjoint operator $A$ has a spectral decomposition
\begin{align}
    A = \sum_i a_i \lvert u_i \rangle \langle u_i \rvert.
\end{align}
The functional calculus is defined as 
\begin{align}
    f(A) \coloneqq \sum_{i:\, a_i \in \text{dom}(f)} f(a_i) \lvert u_i \rangle \langle u_i \rvert.
\end{align}
We adopt the notation $A^0$ to be the orthogonal projection on the $A$'s support; and $A^0 = 0$ for $A = 0$.
We use $\I$ for the identity operator.

For $M$ positive semidefinite trace-class operators \((A_1,\dots,A_M)\)
and a POVM $\{\Lambda_i\}_{i=1}^M$, we define the associated error:
\begin{align}
\Err \left(A_1,\dots,A_M;\{\Lambda_i\}_{i=1}^M\right)
&\coloneqq \sum_{i=1}^M  \Tr\left[ A_i ( \I - \Lambda_i)\right],
\end{align}
and the optimal error by minimizing all POVMs:
\begin{align} \label{eq:Err}
\Err^\star \left(A_1,\dots,A_M\right)
\coloneqq \inf_{\mathrm{POVM}\, \{\Lambda_i\}_{i=1}^M} \Err \left(A_1,\dots,A_M;\{\Lambda_i\}_{i=1}^M\right).
\end{align}
If \(\{\Lambda_0,\Lambda_1,\dots,\Lambda_M\}\) is a POVM with an additional failure outcome \(\Lambda_0\), we define
\begin{align}
\Err\left(A_1,\dots,A_M; \{\Lambda_i\}_{i=0}^M\right)
\coloneqq 
\sum_{i=1}^M  \Tr\left[ A_i ( \I - \Lambda_i)\right].
\end{align}
That is, we count the ``no-detection event'' as an erroneous event.
It is straightforward to see that
\begin{align}
\Err^\star \left(A_1,\dots,A_M\right)
= \inf_{\mathrm{POVM}\, \{\Lambda_i\}_{i=0}^M} \Err \left(A_1,\dots,A_M;\{\Lambda_i\}_{i=0}^M\right).
\end{align}

In binary hypothesis testing, Holevo and Helstrom showed that \eqref{eq:Err} admits a closed-form expression \cite{Hel67, Hol72, Hol74, Hol76}:
\begin{align} \label{eq:HH}
    \Err^{\star}(A_1, A_2)
    = \Tr\left[ A_1 \wedge A_2 \right].
\end{align}
Here, for self-adjoint trace-class operators $A$ and $B$, we define the noncommutative minimum and maximum, respectively, as
\begin{align}
     A\wedge B
     &\coloneqq \argmax_{H=H^\dagger} \left\{ \Tr[H] : H \leq A, H \leq B \right\}
   = \frac{A+B-|A-B|}{2},
     \\
     A\vee B
     &\coloneqq \argmin_{H=H^\dagger} \left\{ \Tr[H] : H \geq A, H \geq B \right\}
     = \frac{A+B+|A-B|}{2}.
\end{align}

\begin{definition}[\NS distributions] \label{defn:NS}
Let $A, B$ be self-adjoint trace-class operators with spectral decompositions:
\begin{align}
    A &= \sum_i a_i |u_i\rangle\langle u_i|,
    \quad
    B = \sum_j b_j |v_j\rangle\langle v_j|.
\end{align}
The (non-normalized) \NS distributions $(P,Q)$ are defined as
\begin{align} \label{eq:defn:NS}
    P
    &\coloneqq \sum_{i,j} a_i |\langle u_i | v_j \rangle |^2 |ij\rangle\langle ij|,
    \quad
    Q
    \coloneqq \sum_{i,j} b_j |\langle u_i | v_j \rangle |^2 |ij\rangle\langle ij|.
\end{align}
\end{definition}
Notice that the \NS distributions are usually defined as $P(i,j) = a_i |\langle u_i | v_j \rangle |^2$ and $Q(i,j) = b_j |\langle u_i | v_j \rangle |^2$ for each $(i,j)$.
In \eqref{eq:defn:NS}, we simply embed $P$ and $Q$ into diagonal operators for notational convenience.

\begin{fact}[Properties for the \NS distributions {\cite{NS09}}] \label{fact:NS} 
Consider self-adjoint trace-class operators $A$ and $B$ and the associated \NS distributions $(P,Q)$ in Definition~\ref{defn:NS}.
The following hold.
    \begin{enumerate}[(i)]
        \item 
        (Non-negativity) $A,B \geq 0$ implies $P,Q \geq 0$.

        \item
        (Support relation) $\supp(A) \subseteq \supp(B)$ implies $\supp(P) \subseteq \supp(Q)$.

        \item 
        (Trace preserving) $\Tr[P] = \Tr[A]$ and $\Tr[Q] = \Tr[B]$.

        \item
        (Tensor factorization)
        If $A = A_1 \otimes A_2$ and $B = B_1 \otimes B_2$ for some self-adjoint trace-class operators $A_i$ and $B_i$, then
        $P = P_1 \otimes P_2$ and $Q = Q_1 \otimes Q_2$, where $(P_i, Q_i)$ are the \NS distributions corresponding to $(A_i, B_i)$.

        \item 
        (Representing Petz's R\'enyi divergence)
        For non-orthogonal $A,B \geq 0$, let \cite{Pet86}
        \begin{align}
        D_{\alpha}(A\Vert B)
        \coloneqq \frac{1}{\alpha-1} \log \Tr\left[ A^{\alpha} B^{1-\alpha}\right], \quad \alpha \in (0,1),
        \end{align}
        and the limits at $\alpha = \{0,1\}$ are defined by the continuous extensions. 
        Note that 
        \begin{align}
            D_1(A\Vert B) = D(A\Vert B) \coloneqq \Tr\left[ A \left( \log(A) - \log(B)\right)\right]
        \end{align}
        if $\supp(A)\subseteq\supp(B)$, and $D(A\Vert B) \coloneqq \infty$ otherwise.
        Then, $D_{\alpha}(P\Vert Q) = D_{\alpha}(A\Vert B)$ for all $\alpha \in [0,1]$.
    \end{enumerate}
\end{fact}

\section{Multiple Hypothesis Testing} \label{sec:multiple}
In this section, we derive the following pairwise upper bound for multiple quantum hypothesis testing.
\begin{theorem}[Pairwise bound] \label{theorem:pairwise}
Let \(A_1,\dots,A_M\) be $M$ positive semidefinite trace-class operators.
For any pair $(A_i,A_j)$, let  
$\big(P^{ij}, Q^{ij}\big)$ be the associated \NS distributions given in Definition~\ref{defn:NS}.
Then,
\begin{align} \label{eq:pairwise00}
\Err^\star \left(A_1,\dots,A_M\right)
\leq
4\sum_{(i,j):\, 1\le i<j\le M}
\Err^\star \!\big(P^{ij}, Q^{ij}\big)
\leq
\begin{dcases}
8\sum_{(i,j):\, 1\le i<j\le M}
\Err^\star \left(A_i,A_j\right),
\\
4\sum_{(i,j):\, 1\le i<j\le M} \Tr\left[A_i^{1-s_{ij}} A_j^{s_{ij}}\right], \quad \forall\, s_{ij}\in(0,1).
\end{dcases}
\end{align}

Moreover, for any number of copies $n$ and any quantum ensemble $(p_i, \rho_i^{\otimes n})_{i=1}^M$,
\begin{align} \label{eq:pairwise01}
\Err^\star \left(p_1 \rho_1^{\otimes n},\dots, p_M \rho_M^{\otimes n}\right)
\leq 4 (M-1) \e^{-n C(\rho_1, \ldots, \rho_M)  },
\end{align}
where the multiple Chernoff distance is defined as $C(\rho_1,\ldots,\rho_M) \coloneqq \min_{i\neq j} \sup_{\alpha \in (0,1) } - \log \Tr\left[\rho_i^{\alpha} \rho_j^{1-\alpha}\right]$.
\end{theorem}
\begin{remark}[Comparisons]
    Compared to the landmark finite-dimensional multiple quantum Chernoff bound by Li \cite[Theorem~2]{Li16}, \eqref{eq:pairwise00} of Theorem~\ref{theorem:pairwise} improves upon the state-dependent prefactor from $\frac{25}{4}(M-1)^2 \max_{1 \le i \le M} \nu^2(p_i \rho_i) + 3$ to the universal constant $4$. Here, $\nu(A)$ denotes the number of distinct eigenvalues of $A$.

    Compared with the best classical result known to us, the only
    additional loss in our pairwise reduction is the universal factor \(4\),
    which arises from the quantum union bound in
    Theorem~\ref{theorem:Union_Bound}; see \cite[Theorem~15]{SV18}.
\end{remark}

\begin{proof}
The proof is divided into three steps.

\noindent\textbf{Step 1: Fine-Graining.}
For each \(i\), choose a rank-one spectral decomposition
\begin{align}
A_i=\sum_{k=1}^{\infty}\lambda_{ik}Q_{ik},
\end{align}
where \(\lambda_{ik}>0\) are the nonzero eigenvalues of \(A_i\), counted with multiplicity, and \(Q_{ik}\) are rank-one orthogonal projections. The convergence is in trace norm.

We enumerate all rank-one spectral components of
\(A_1,\ldots,A_M\) as a single sequence
\begin{align}
    (\lambda_1,Q_1),(\lambda_2,Q_2),\ldots
\end{align}
arranged in non-increasing order of the eigenvalues:
\begin{align}
    \lambda_1\ge \lambda_2\ge \cdots .
\end{align}
For each \(m\in\mathbb N\), let \(i(m)\in\{1,\dots,M\}\) be the index such that
\((\lambda_m,Q_m)\) comes from the spectral decomposition of \(A_{i(m)}\).
Since the spectral expansion of each \(A_i\) is absolutely convergent in trace norm, the relabeling does not affect the value of the sum. Hence, for each \(i\),
\begin{align}
A_i=\sum_{m:\,i(m)=i}\lambda_mQ_m
\end{align}
with unconditional convergence in trace norm.

We first reduce the problem to discriminating the finer family
\[
    \lambda_1Q_1,\lambda_2Q_2,\ldots .
\]
To this end, consider an arbitrary POVM
\(\{E_m\}_{m=0}^{\infty}\) for the fine-grained problem,
where the outcome \(m\ge 1\) is interpreted as deciding the fine
hypothesis \(\lambda_m Q_m\), and the outcome \(0\) is a residual outcome.

From this POVM, we obtain a POVM with \(M\) decision outcomes and one residual outcome by merging the fine outcomes according to the original labels:
\begin{align}
\Lambda_i
\coloneqq
\sum_{m:\,i(m)=i}E_m,
\qquad i=1,\dots,M,
\quad\text{and}\quad
\Lambda_0
\coloneqq
E_0,
\end{align}
where the sums converge in the strong operator topology (SOT), since their partial sums are monotonically increasing and bounded by \(\I\). Then
\(\{\Lambda_0,\Lambda_1,\dots,\Lambda_M\}\) is a POVM. The residual outcome is counted as an error in the estimate below; equivalently, assigning it to any one of the original labels does not increase the error.

Moreover, since \(\Lambda_i\) is defined as the strong limit of the
increasing finite partial sums of positive operators \(E_m\), it follows that
\(\Lambda_i \ge E_m\) for every \(m\) such that \(i(m)=i\). Since \(A_i=\sum_{m:i(m)=i}\lambda_mQ_m\) in trace norm and
\(1-\Lambda_i\) is bounded, the continuity of the trace pairing yields
\begin{align}
\sum_{i=1}^M\Tr[A_i(\I-\Lambda_i)]
&=
\sum_{i=1}^M
\sum_{m:\,i(m)=i}
\lambda_m\Tr[Q_m(\I-\Lambda_i)] \notag\\
&\le
\sum_{i=1}^M
\sum_{m:\,i(m)=i}
\lambda_m\Tr[Q_m(\I-E_m)] \notag\\
&=
\sum_{m=1}^{\infty}\lambda_m\Tr[Q_m(\I-E_m)].
\end{align}
Thus,
\begin{align}
\Err^\star(A_1,\dots,A_M)
\le
\Err\left(\lambda_1Q_1,\lambda_2Q_2,\dots;\{E_m\}_{m=0}^{\infty}\right),
\end{align}
where
\begin{align}
\Err\left(\lambda_1Q_1,\lambda_2Q_2,\dots;\{E_m\}_{m=0}^{\infty}\right)
\coloneqq
\sum_{m=1}^{\infty}\lambda_m\Tr[Q_m(\I-E_m)].
\end{align}
In other words, it suffices to construct a good measurement for the finer problem.

\medskip
\noindent\textbf{Step 2: POVM Construction.}
Set
\begin{align}
\bar Q_m\coloneqq \I-Q_m,
\qquad m\in\mathbb N.
\end{align}
We now construct such a POVM by sequentially applying the two-outcome
projective measurements \(\{Q_m,\bar Q_m\}\) in the order
\(m=1,2,\ldots\). At step \(m\), the measurement is performed only if none
of the previous projections \(Q_1,\ldots,Q_{m-1}\) has occurred. The effect
corresponding to the first occurrence of \(Q_m\) is therefore
\begin{align}
E_m
\coloneqq
\bar Q_1\cdots \bar Q_{m-1}Q_m\bar Q_{m-1}\cdots \bar Q_1,
\qquad m\in\mathbb N.
\end{align}

The residual outcome corresponds to the event that none of the projections
\(Q_1,Q_2,\ldots\) is accepted. To define its effect, let
\begin{align}
L_m
\coloneqq
\bar Q_1\cdots \bar Q_{m-1} \bar Q_m \bar Q_{m-1}\cdots \bar Q_1,
\qquad m\in\mathbb N,
\end{align}
with \(L_0=\I\). Here, \(L_m\) is the effect corresponding to rejecting
\(Q_1,\ldots,Q_m\). Since \(E_m\ge0\), the identity
\begin{align}
L_{m-1}=E_m+L_m
\end{align}
shows that \((L_m)_{m\ge0}\) is decreasing in the operator order. Being
bounded below by \(0\), it therefore converges strongly to some positive
operator. We define the residual effect by
\[
    E_0\coloneqq \operatorname*{s-lim}_{m\to\infty}L_m ,
\]
which corresponds to rejecting all projections \(Q_1,Q_2,\ldots\).

For every \(n\in\mathbb N\), iterating \(L_{m-1}=E_m+L_m\) gives
\begin{align}
\I=L_0=\sum_{m=1}^{n}E_m+L_n.
\end{align}
Since \(L_n\to E_0\) in the SOT, we have
\begin{align}
\sum_{m=1}^{n}E_m
=
\I-L_n
\xrightarrow[n\to\infty]{\mathrm{SOT}}
\I-E_0.
\end{align}
Therefore,
\begin{align}
\I=\sum_{m=1}^{\infty}E_m+E_0
\end{align}
in the strong operator topology. Thus \(\{E_0,E_1,E_2,\dots\}\) is a POVM.

\medskip
\noindent\textbf{Step 3: Error Analysis.}
It remains to estimate this fine error. Fix \(m \in \mathbb{N}\). Since
\begin{align}
E_m=
\bar Q_1\cdots \bar Q_{m-1}Q_m\bar Q_{m-1}\cdots \bar Q_1,
\end{align}
we have
\begin{align}
\lambda_m\Tr[Q_mE_m]
=
\Tr\left[
Q_m\bar Q_{m-1}\cdots \bar Q_1
(\lambda_mQ_m)
\bar Q_1\cdots \bar Q_{m-1}Q_m
\right].
\end{align}
Apply the quantum union bound, Theorem~\ref{theorem:Union_Bound} in Appendix~\ref{sec:union_bound}, with
\begin{align}
\rho\leftarrow \lambda_mQ_m,
\quad
\Pi_1\leftarrow \bar Q_1,\ \dots,\ \Pi_{m-1}\leftarrow \bar Q_{m-1},
\quad
\Pi_m\leftarrow Q_m.
\end{align}
Since \(\I-\bar Q_v=Q_v\) for \(v<m\), and
\begin{align}
(\I-Q_m)(\lambda_mQ_m)=0,
\end{align}
we obtain
\begin{align}
\lambda_m\Tr[Q_mE_m]
&\ge
\Tr[\lambda_mQ_m]
-
4\sum_{v:\,v<m}\Tr[Q_v(\lambda_mQ_m)] \notag\\
&=
\lambda_m
-
4\sum_{v:\,v<m}\lambda_m\Tr[Q_vQ_m].
\end{align}
Therefore,
\begin{align}
\Err\left(\lambda_1Q_1,\lambda_2Q_2,\dots;\{E_m\}_{m=0}^{\infty}\right)
&=
\sum_{m=1}^{\infty}\lambda_m
(1-
\Tr[Q_mE_m]) \notag\\
&\le
4\sum_{m=1}^{\infty}
\sum_{v:\,v<m}
\lambda_m\Tr[Q_vQ_m].
\end{align}
Because \(\lambda_1\ge\lambda_2\ge\cdots\), for every \(v<m\),
\begin{align}
\lambda_m=\lambda_v\wedge\lambda_m.
\end{align}
Hence, by the preceding fine-graining reduction,
\begin{align}
\Err^\star(A_1,\dots,A_M)
\le
4\sum_{m=1}^{\infty}
\sum_{v:\,v<m}
(\lambda_v\wedge\lambda_m)\Tr[Q_vQ_m].
\end{align}

Finally, rewriting this sum in terms of the original spectral decompositions gives
\begin{align}
\sum_{m=1}^{\infty}
\sum_{v:\,v<m}
(\lambda_v\wedge\lambda_m)\Tr[Q_vQ_m]
=
\sum_{(i,j):1\le i<j\le M}
\sum_{k=1}^{\infty}
\sum_{\ell=1}^{\infty}
(\lambda_{ik}\wedge\lambda_{j\ell})
\Tr[Q_{ik}Q_{j\ell}],
\end{align}
because for fixed \(i\), the projections \(\{Q_{ik}\}_k\) are pairwise orthogonal. 
As a result,
\begin{align}
\Err^\star(A_1,\dots,A_M)
\le
4\sum_{(i,j):1\le i<j\le M}
\sum_{k=1}^{\infty}
\sum_{\ell=1}^{\infty}
(\lambda_{ik}\wedge\lambda_{j\ell})
\Tr[Q_{ik}Q_{j\ell}].
\end{align}
For each pair \((i,j)\), the summation is precisely
\(\Err^\star(P^{ij},Q^{ij})=\Tr[P^{ij}\wedge Q^{ij}]\) for the associated \NS distributions. This proves the first claim of \eqref{eq:pairwise00}.

The inequality
\begin{align}
\Tr[P^{ij}\wedge Q^{ij}]
\le
2\Tr[A_i\wedge A_j]
\end{align}
is the \NS lower bound \eqref{eq:NS_lower} discussed in Appendix~\ref{sec:NS}. The Chernoff bound follows from the scalar inequality
\(x\wedge y\le x^{1-s}y^s\), \(x,y\ge 0\), \(s\in(0,1)\):
\begin{align}
\Tr[P^{ij}\wedge Q^{ij}]
&\le
\Tr\left[\big(P^{ij}\big)^{1-s}\big(Q^{ij}\big)^s\right] \notag\\
&=
\sum_{k,\ell}
\lambda_{ik}^{1-s}\lambda_{j\ell}^{s}
\Tr[Q_{ik}Q_{j\ell}] \notag\\
&=
\Tr\left[A_i^{1-s}A_j^s\right],
\end{align}
where the last line follows from Fact~\ref{fact:NS}.
In the i.i.d.~setting, for any \(n\in\mathbb N\), applying the preceding
scalar Chernoff bound and then optimizing over each pair-dependent
parameter \(s_{ij}\in(0,1)\) gives:
\begin{align}
    \Err^\star(p_1\rho_1^{\otimes n}, \dots, p_M\rho_M^{\otimes n}) 
    &\le 4 \sum_{(i,j):\, 1 \le i < j \le M} \inf_{s_{ij}\in(0,1)}p_i^{1-s_{ij}} p_j^{s_{ij}} \left(\Tr[\rho_i^{1-s_{ij}} \rho_j^{s_{ij}}]\right)^n \\
    &\le 4 \sum_{(i,j):\, 1 \le i < j \le M} (p_i+p_j)\inf_{s_{ij}\in(0,1)} \left(\Tr[\rho_i^{1-s_{ij}} \rho_j^{s_{ij}}]\right)^n \\
    &= 4 \sum_{(i,j):\, 1 \le i < j \le M} (p_i+p_j) \cdot \e^{-n C(\rho_i, \rho_j)} 
    \\
    &\le 4 \sum_{(i,j):\, 1 \le i < j \le M} (p_i+p_j) \cdot \e^{-n C(\rho_1, \ldots, \rho_M)} \label{eq:last1}
    \\
    &= 4 (M-1) \cdot \e^{-n C(\rho_1, \ldots, \rho_M)},
\end{align}
concluding the proof.
\end{proof}

\section{Binary Hypothesis Testing and Harmonic-Mean Bounds} \label{sec:simple}

The pairwise bound established in Theorem~\ref{theorem:pairwise} directly applies to binary hypothesis testing ($M=2$):
\begin{align} \label{eq:loose_bound}
    \Err^\star\left(A,B\right)
    \leq 4 \cdot \Err^\star\left(P,Q\right),
\end{align}
where $(P,Q)$ are the \NS distributions of $(A,B)$.
In this section, we will show that the factor $4$ in \eqref{eq:loose_bound} can be further strengthened to $2$ via the harmonic mean as a bridge.
Together with the \NS lower bound \cite{NS09} (see also Appendix~\ref{sec:NS}), we obtain the upper and lower control of the quantum error by the classical error, i.e.,
\begin{align} \label{eq:same_order}
    \frac12 \cdot \Err^\star\left(P,Q\right)
    \leq
    \Err^\star\left(A,B\right)
    \leq 2 \cdot \Err^\star\left(P,Q\right).
\end{align}

We first formally introduce the definition of the harmonic mean and present our main result---harmonic-mean bound---in Theorem~\ref{theorem:harmonic-mean_bound} below.
For positive semidefinite operators $A,B\geq 0$ and parameter $s \in [0,1]$, the \emph{$s$-weighted Kubo--Ando harmonic operator mean} \cite{KA80} is defined as
\begin{align} \label{eq:KA_harmonic_mean}
    A \mathbin{!}_s B 
    &\coloneqq A \left( s A + (1-s) B\right)^{-1} B
    = \left( (1-s)A^{-1} + s B^{-1} \right)^{-1}
\end{align}
with the standard limiting definition for singular $A$  and $B$.\footnote{When $A$ or $B$ is singular, the Kubo--Ando harmonic operator mean in \eqref{eq:KA_harmonic_mean} is defined as $A \mathbin{!}_s B 
\coloneqq \lim_{\epsilon \searrow 0} \big(A + \epsilon \I\big) \mathbin{!}_s \big(B + \epsilon \I\big)$, where the limit is in the taken in the operator order  and the strong operator topology.
}

\begin{definition}[Petz--Nussbaum--Szko{\l}a harmonic mean] \label{defn:harmonic_mean}
    Let $A$ and $B$ be positive semidefinite trace-class operators with the associated \NS distributions $(P,Q)$ given in Definition~\ref{defn:NS}.
    For any $s\in [0,1]$, the \emph{$s$-weighted Petz--\NS harmonic mean} of $A$ and $B$ is defined as\footnote{By convention, the fraction in \eqref{eq:defn:HM} is understood as $\lim_{\epsilon\searrow 0} \frac{(a_i+\epsilon) (b_j+\epsilon) }{s (a_i+\epsilon) + (1-s) (b_j+\epsilon)} $.}
    \begin{align} \label{eq:defn:HM}
        \HM_s (A,B)
        \coloneqq \Tr[P\mathbin{!}_s Q]
        &= \sum_{i,j} \frac{a_i b_j}{s a_i + (1-s) b_j} |\langle u_i | v_j \rangle|^2.
    \end{align}
\end{definition}

The quantity in \eqref{eq:defn:HM} is monotonically non-decreasing under any completely positive and trace-preserving (CPTP) map $\mathscr{N}$, i.e.,
\begin{align}
    \HM_s (A,B) \leq \HM_s \left(\mathscr{N}(A),\mathscr{N}(B)\right).
\end{align}
Indeed, $\HM_s(A,B)$ is within a broader family of quantum metrics studied by Petz~\cite{Pet85, Pet86, petz1996monotone} (see also \cite{Lesniewski1999, HMP+11, Hia18}) via evaluating an operator monotone function on the relative modular operator $\Delta_{{A}|{B}} := A (\cdot ) B^{-1}$ \cite{araki1975relative, araki1977relative}:\footnote{When $B$ is singular, the right multiplication in $\Delta_{A\mid B}$ is taken on the support of $B$, i.e.~pseudoinverse.}\textsuperscript{,}\footnote{The second term $f_s'(\infty) \Tr\left[ A ( \I - B^0) \right]$ in \eqref{eq:Petz_form} is to correct the end point $s=0$, i.e., $f_s'(\infty) = 1$ so that the convention $M_0(A,B) = \Tr[A]$ (instead of $\Tr[A B^0]$) is preserved.
However, the consideration of the end points $s = 0,1$ is immaterial to this work.
}
\begin{align} \label{eq:Petz_form}
        \HM_s (A,B) 
        = \big\langle B^{\nicefrac12} , f_s\left(\Delta_{A\mid B}\right) B^{\nicefrac12} \big\rangle  + f_s'(\infty) \Tr\left[ A ( \I - B^0) \right],
\end{align}
where the generating function for the weighted harmonic mean is the following operator monotone function \cite[\S 4]{HP14}:
\begin{align} \label{eq:generating_function}
    f_s:  [0,+\infty)\ni x\mapsto \frac{x}{s x + (1-s)}, \quad s \in [0,1].
\end{align}
We collect other useful representations of $\HM_s(A,B)$ in Appendix~\ref{sec:representation}.

\begin{theorem}[Harmonic-mean bound] \label{theorem:harmonic-mean_bound}
	For all trace-class operators $A, B \geq 0$ with the \NS distributions $(P,Q)$ and for all parameters $s \in (0,1)$, the following chain of inequalities holds:
	\begin{align}
        s\wedge (1-s) \cdot \Tr[P\wedge Q]
		\leq s\wedge (1-s) \cdot \Tr[P\mathbin{!}_s Q] 
		\leq \Tr[A \wedge B]
        \leq \Tr[P\mathbin{!}_s Q] 
        \leq
        \begin{dcases}
        \frac{1}{s\wedge (1-s)} \Tr[P\wedge Q],
        \\
        \Tr\left[A^{1-s} B^s\right].
        \end{dcases}
	\end{align}
\end{theorem}

The lower bounds $s\wedge (1-s) \cdot \Tr[P\wedge Q] \leq s\wedge (1-s) \cdot \Tr[P\mathbin{!}_s Q] \leq \Tr[A \wedge B]$ are discussed in Appendix~\ref{sec:NS}.
The harmonic-mean upper bound $\Tr[A\wedge B] \leq \Tr[P \mathbin{!}_{s} Q] \leq \Tr[A^{1-s} B^{s}]$ provides an alternative proof to the quantum Chernoff bound \cite[Theorem 1]{ACM+07}, \cite[Theorem 2]{ANS+08}.

\begin{remark}[Rank-one operators]
When $A$ and $B$ are rank-one operators, one can upper bound $\Tr[A\wedge B]$ by $\Tr[P\wedge Q]$ directly without going through the harmonic mean $\Tr[P\mathbin{!}_s Q]$.
For example, for pure states $|\phi\rangle \langle \phi|$ and $|\psi\rangle \langle \psi|$ with overlap $|\langle\phi|\psi\rangle|^2 = c$ and with the associated \NS distributions $(P,Q)$, a direct calculation shows that, for all $\gamma \geq 0$,
\begin{align}
    \Tr[|\phi\rangle \langle \phi|\wedge \gamma |\psi\rangle \langle \psi| ]
    = \frac{1}{2} \left( 1 + \gamma - \sqrt{(1+\gamma)^2 - 4\gamma c} \right) \leq c \cdot \min(1, \gamma) = \Tr[P\wedge \gamma Q].
\end{align}

In general, $\Tr[\rho \wedge \gamma \sigma] \not\leq \Tr[P \wedge \gamma Q]$.
For example, consider $\rho = |0\rangle \langle 0|$ and $\sigma = (1-\epsilon) |+\rangle\langle +| + \epsilon |-\rangle \langle -|$ with $\epsilon = \frac{1}{2\gamma}$.
Then, for $\gamma \geq \nicefrac23$,
\begin{align}
    \frac{ \Tr[\rho \wedge \gamma \sigma] }{ \Tr[P \wedge \gamma Q] } 
    = \frac{2}{3} \left( 1 + \gamma - \sqrt{(\gamma-1)^2+1}\right)
    \overset{\gamma \nearrow +\infty}{\longrightarrow} \frac{4}{3}.
\end{align}
It is, however, not clear if the factor $\frac{4}{3}$ is universal.
\end{remark}


\begin{remark}[Other harmonic means]
We discuss four variants of harmonic means subsequently.
\begin{enumerate}[1.]
    \item
    (Kubo--Ando harmonic operator mean).
    It is well-known that the tracial Kubo--Ando harmonic mean admits the following Belavkin--Staszewski expression \cite{BS82}:
    \begin{align}
    \Tr[A\mathbin{!}_s B] = \Tr\left[ B^{\nicefrac12} f_{s}\left( B^{-\nicefrac12} A B^{-\nicefrac12} \right) B^{\nicefrac12} \right] + f_s'(\infty) \Tr\left[ A ( \I - B^0) \right]
    \end{align}
    with the operator monotone function $f_{s}$ given in \eqref{eq:generating_function}.
    Moreover, $\Tr[A\mathbin{!}_s B]$ is the \emph{minimal} harmonic mean among all quantum harmonic means obeying monotonicity \cite{HM17, Matsumoto, Hia19}.
    However, it could be too small to upper bound $\Tr[A\wedge B]$.
    Indeed, for non-orthogonal pure states $|\phi\rangle$ and $|\psi\rangle$ with no common support, we have
    \begin{align}
        0 < \Tr\left[|\phi\rangle \langle \phi| \wedge |\psi\rangle \langle \psi | \right] 
        \not\leq 
        \Tr\left[|\phi\rangle \langle \phi| \mathbin{!}_s |\psi\rangle \langle \psi | \right] = \Tr[0] = 0,
        \quad \forall\, s \in (0,1).
    \end{align} 

    \item
    (Measured harmonic mean).
    One can consider the minimal harmonic mean of measurement outcomes:
    \begin{align}
        \inf_{ \text{POVM } \{M_z\}_z } \sum_z \Tr[M_z A] \mathbin{!}_s \Tr[M_z B].
    \end{align}
    By the scalar inequality $x\wedge y \leq x \mathbin{!}_s y$, the measured harmonic mean immediately upper bounds $ \Tr[A\wedge B] = \inf_{ \text{POVM } \{M_z\}_z } \sum_z \Tr[M_z A] \wedge \Tr[M_z B] $.
    Moreover, it is \emph{maximal} among all quantum harmonic means obeying monotonicity \cite{HM17}.

    \item 
    (Standard pretty-good measurement).
    Using the standard two-outcome pretty-good measurement \cite{Bel75, HW94},
    one can bound
    \begin{align}
    \Tr[A\wedge B] 
    \leq \Tr[A \left( s A + (1-s) B \right)^{-\nicefrac12} B \left( s A + (1-s) B \right)^{-\nicefrac12}].
    \end{align}
    The quantity on the right-hand side is also monotonically non-decreasing under any CPTP map (see e.g.~\cite[Appendix]{Cheng_simple}); however, it is numerically incomparable with $\HM_s (A,B) = \Tr[P\mathbin{!}_s Q]$.

    \item 
    (Integral pretty-good measurement).
    Using the integral two-outcome pretty-good measurement \cite{BT24, preparation},
    one can bound
    \begin{align}
    \Tr[A\wedge B] 
    \leq
    \int_0^\infty \Tr\left[ A \left( s A + (1-s) B + t \I \right)^{-1} B \left( s A + (1-s) B + t \I \right)^{-1} \right] \d t.
\end{align}
    The quantity on the right-hand side is also monotonically non-decreasing under any CPTP map as it is closely related to the quantum Le Cam divergence \cite{hirche2023quantum, beigi2025some, LHC25_layer_cake}; however, it is again numerically incomparable with $\HM_s (A,B) = \Tr[P\mathbin{!}_s Q]$.
    
\end{enumerate}
\end{remark}

\begin{remark}[Relation to Chernoff]
The quantum Chernoff bound $\Tr[A^{1-s}B^s]$ is an immediate consequence of the harmonic-mean bound and the scalar harmonic--geometric mean inequality.
It is, however, nontrivial to conversely control the Petz--\NS harmonic mean via $\Tr[A^{1-s}B^s]$.
Indeed, using the integral representation, $x^{1-s} y^s = \frac{\sin(s\pi)}{\pi} \int_0^1 t^{s-1} (1-t)^{-s} x \mathbin{!}_t y \, \d t$,
\begin{align}
    \Tr[A^{1-s} B^s]
    = \mathds{E}_{t \sim \mu_s} \HM_t(A,B)
    \geq \HM_{\mathds{E}_{t\sim \mu_s}[t]}(A,B)
    = \HM_s (A,B), \quad \forall\, s\in(0,1),
\end{align}
where the probability measure is $\mu_s(t) = \frac{\sin(s\pi)}{\pi} t^{s-1} (1-t)^{-s}$ with the expectation $\mathds{E}_{t\sim \mu_s}[t] = s$, and the map $t\mapsto \HM_t(A,B)$ is convex. 
This means that $\Tr[A^{1-s}B^s] \leq \sup_{t\in (0,1)} \HM_t(A,B) = \Tr[AB^0] \vee \Tr[B A^0]$, but there is no universal constant $c_{s,t}$ such that $\Tr[A^{1-s}B^s] \leq c_{s,t} \HM_t(A,B)$.
\end{remark}

\begin{proof}[Proof of Theorem~\ref{theorem:harmonic-mean_bound}]
We start with the supremum representation \eqref{eq:supremum_representation} given in Proposition~\ref{prop:representations}. 
Choosing a feasible operator $Y_0 = A (A\vee B)^{-1} B$, we have
\begin{align} \label{eq:proof_harmonic-mean_1}
H_{s}(A,B) = \Tr[P\mathbin{!}_s Q]
    \geq \Tr[Y_0+Y_0^\dagger] - (1-s) \Tr[Y_0^\dagger A^{-1} Y_0] - s\Tr[ Y_0 B^{-1} Y_0^{\dagger} ] . 
\end{align}
We will show that the right-hand side of \eqref{eq:proof_harmonic-mean_1} upper bounds $\Tr[A\wedge B]$.
To that end, we invoke Lemma~\ref{lemma:min_identity} below to obtain $\textsc{Min} \coloneqq \Tr[A\wedge B] = \Tr[Y_0]$. 
Also, $\Tr[Y_0^\dagger] = \Tr[ B (A\vee B)^{-1} A ] = \Tr[Y_0] = \textsc{Min}$.
Defining the quadratic terms 
$q_A = \Tr[Y_0^\dagger A^{-1} Y_0]$
and $q_B = \Tr[Y_0 B^{-1} Y_0^\dagger]$, the right-hand side of \eqref{eq:proof_harmonic-mean_1} is now 
\begin{align} \label{eq:proof_harmonic-mean_2}
    2 \textsc{Min} - (1-s) q_A - s q_B.
\end{align}

Using the orthogonal components $H_+ = (A-B)_+$ and $H_- = (B-A)_+$ and the identities
\begin{align}
A = A\vee B - H_-, \;
B = A\vee B - H_+, \; \text{and }
A\wedge B = A\vee B - H_+ - H_-,
\end{align}
a direct calculation shows that
\begin{align}
    q_A &= \Tr\left[B(A\vee B)^{-1}B\right] - \Tr\left[B(A\vee B)^{-1}H_-(A\vee B)^{-1}B\right]
    \\
    &= \textsc{Min} - \Tr[H_+ - H_+ (A \vee B)^{-1} H_+] - \Tr[H_+ (A \vee B)^{-1} H_- (A \vee B)^{-1} H_+], 
    \label{eq:q_A}
    \\
    q_B &= \Tr\left[A(A\vee B)^{-1}A\right] - \Tr\left[A(A\vee B)^{-1}H_+(A\vee B)^{-1}A\right]
    \\
    &=\textsc{Min} - \Tr[H_- - H_- (A \vee B)^{-1} H_-] - \Tr[H_- (A \vee B)^{-1} H_+ (A \vee B)^{-1} H_-].
\end{align}
The operator monotonicity of the inverse function (see e.g.~\cite[\S 4]{HP14}) implies that, for any $\epsilon>0$,
\begin{align}
    \left( A\vee B + \epsilon \I \right)^{-1}
    &= \left( H_+ + B + \epsilon \I \right)^{-1}
    \leq \left( H_+ + \epsilon \I \right)^{-1}.
\end{align}
Multiplying both sides by $H_+^{\nicefrac12}$ and taking the limit $\epsilon \searrow 0$, we obtain
\begin{align} \label{eq:H_+_monotone}
    H_+^{\nicefrac12} \left( A\vee B \right)^{-1} H_+^{\nicefrac12}
    \leq H_+^0.
\end{align}
This implies that
\begin{align}
    \Tr[H_+ - H_+ (A \vee B)^{-1} H_+]
    = \Tr\left[ H_+^{\nicefrac12} ( H_+^0 - H_+^{\nicefrac12} \left( A\vee B \right)^{-1} H_+^{\nicefrac12} ) H_+^{\nicefrac12} \right] \geq 0.
\end{align}
Hence, from \eqref{eq:q_A}, we see $q_A \leq \textsc{Min}$, and likewise $q_B \leq \textsc{Min}$.
Combining with \eqref{eq:proof_harmonic-mean_2}, we prove
\begin{align}
    \HM_s(A,B) \geq 
    2\Min - (1-s) q_A - s q_B
    \geq     2\Min - (1-s)\Min - s\Min = \textsc{Min}.
\end{align}

The other upper bounds directly follow from the scalar inequalities $x \mathbin{!}_s y \leq \frac{1}{s \wedge (1-s)} x \wedge y$ and $x \mathbin{!}_s y \leq x^{1-s} y^{s}$.
The lower bounds are proved in Proposition~\ref{prop:NS_lower}.
\end{proof}

\begin{lemma}[Pseudo-Hermitian identity] \label{lemma:min_identity}
Let $A,B\geq 0$ be positive semidefinite trace-class operators.
Then,
\begin{align}
    \Tr\left[A\wedge B\right]
    = \Tr\left[ A (A\vee B)^{-1} B \right].
\end{align}
The right-hand side of the identity is understood as $\lim_{\epsilon \searrow 0} \Tr\left[ A (A\vee B + \epsilon \I)^{-1} B \right]$.
\end{lemma}

\begin{proof}
    We first consider finite-dimensional $A,B>0$ for simplicity; then, $A\vee B > 0$.
    Let $H_+ = (A-B)_+$ and $H_- = (B-A)_+$, such that $|A-B| = H_+ + H_-$ and $H_+ H_- = 0$.
    Moreover, 
    \begin{align}
        A = A\vee B - H_-
        \quad \text{and} \quad
        B = A\vee B - H_+.
    \end{align}
    Substituting these into the target trace, we have
    \begin{align}
        \Tr\left[ A (A\vee B)^{-1} B \right]
        &= \Tr\left[ ( A\vee B - H_- ) (A\vee B)^{-1} ( A\vee B - H_+) \right]
        \\
        &= \Tr\left[ A\vee B - H_+ - H_- + H_- (A\vee B)^{-1} H_+ \right].
    \end{align}
    By the cyclic property of the trace, 
    $\Tr\left[ H_- (A\vee B)^{-1} H_+\right] = \Tr\left[ H_+ H_- (A\vee B)^{-1}\right] = 0$.
    Hence, 
    \begin{align}
        \Tr\left[A (A\vee B)^{-1} B\right] = \Tr\left[A\vee B - (H_+ + H_-)\right] = \Tr\left[ A\vee B - |A-B|\right] = \Tr\left[A\wedge B\right].
    \end{align}

    For general trace-class $A,B\geq 0$, we replace the inverse $(A\vee B)^{-1}$ by the bounded resolvent $(A \vee B + \epsilon \mathbf{1})^{-1}$ for $\epsilon > 0$. 
    Expanding $\Tr\left[ A (A\vee B + \epsilon\mathbf{1})^{-1} B \right]$ analogously, the cross-term $\Tr\left[ H_- (A\vee B + \epsilon\mathbf{1})^{-1} H_+ \right]$ vanishes for each $\epsilon > 0$ via the cyclic property because $H_+ H_- = 0$.
    To pass the limit $\epsilon \searrow 0$ inside the trace, observe that $A \vee B$ is trace-class, and $H_\pm = (A-B)_\pm$ are trace-class because $A-B$ is.
    Furthermore, $(A \vee B)(A \vee B + \epsilon\mathbf{1})^{-1} \to (A \vee B)^0$ is uniformly bounded and SOT-convergent. Hence, the trace-class/bounded duality guarantees the continuity of the trace operation here, allowing us to evaluate:
    \begin{align}
        \lim_{\epsilon \searrow 0} \Tr\left[ A (A\vee B + \epsilon\mathbf{1})^{-1} B \right] 
        &= \Tr\left[ (A\vee B)(A\vee B)^0 - H_+ (A\vee B)^0 - H_- (A\vee B)^0 \right] \nonumber \\
        &\overset{(\star)}{=} \Tr\left[ A\vee B - H_+ - H_- \right] \nonumber \\
        &= \Tr\left[A\wedge B\right],
    \end{align}
    where $(\star)$ follows because $H_\pm \leq A \vee B$, meaning the support projection $(A \vee B)^0$ acts as the identity on their respective supports. This validates the extension to arbitrary trace-class operators $A,B$ and completes the proof.
\end{proof}

\section{Almost-Exact Asymptotics} \label{sec:asymptotics}

A crucial advantage of the harmonic-mean bound in Theorem~\ref{theorem:harmonic-mean_bound} is to compare the minimum error $\Tr[A\wedge B]$ for binary quantum hypothesis testing with 
that of the associated \NS distributions with only a factor of $2$ loss (i.e.~choosing $s = \nicefrac12$).
Hence, such a one-shot bound is tight enough to provide an almost-exact i.i.d.~asymptotics of quantum hypothesis testing via the classical strong large deviation theory.

We assume the standard Cram\'er's conditions: the non-orthogonal states $\rho$ and $\sigma$ are not identical on the common support, strict interiority condition, and the associated \NS distributions $(P,Q)$ are non-lattice.

\begin{fact}[Cram\'er's condition, symmetric setting {\cite{BR60}, \cite{Petrov1965}, \cite[Chapters VII, VIII]{Petrov1975}  \cite[\S 2.2]{DZ98}}] \label{fact:condition_symmetric}
Suppose the pair of density operators $(\rho,\sigma)$ and the corresponding \NS distributions $(P,Q)$ satisfy
\begin{subequations} \label{eq:assumption_symmetric}
\begin{align}
    &\text{(Non-orthogonality)} &&P \not\perp Q,
    \label{eq:assumption_symmetric1}
    \\
    &\text{(Non-degeneracy)}  &&\text{$P \not\propto Q$ on their common support,} 
    \label{eq:assumption_symmetric2}
    \\
    &\text{(Non-lattice)} &&\text{$(P,Q)$ are non-lattice.}
    \label{eq:assumption_symmetric3}
    \\
    &\text{(Interiority)} &&
    \left.\frac{\d}{\d\alpha}\log \Tr\left[\rho^{\alpha}\sigma^{1-\alpha}\right]\right|_{\alpha=0^+}<0
    \quad\text{and}\quad
    \left.\frac{\d}{\d\alpha}\log \Tr\left[\rho^{\alpha}\sigma^{1-\alpha}\right]\right|_{\alpha=1^-}>0.
\end{align}
\end{subequations}
Then, there exists a unique $\alpha^{\star} \in (0,1)$ satisfying
\begin{align}
&C(\rho, \sigma) = - \log \Tr\left[\rho^{\alpha^{\star}} \sigma^{1-\alpha^{\star}}\right]
\\
&V_{\alpha^\star}(\rho\Vert\sigma)
\coloneqq \textnormal{Var}_{P_{\alpha^\star}}\left[ \log \frac{P}{Q} \right] \in (0,+\infty),
\label{eq:V}
\end{align}
where the tilted distribution for $\alpha \in [0,1]$ is defined as
\begin{align}
    P_{\alpha} \coloneqq \frac{P^{\alpha} Q^{1-\alpha}}{\Tr\left[P^{\alpha} Q^{1-\alpha}\right]}.
\end{align}
\end{fact}

\begin{theorem}[Almost-exact asymptotics for symmetric hypothesis testing]\label{theorem:asymptotics_symmetric}
    For all density operators $\rho$ and $\sigma$ and the corresponding \NS distributions $(P,Q)$
    satisfying ~\eqref{eq:assumption_symmetric}, and for all prior probability $p\in (0,1)$ and $\bar{p} \coloneqq 1-p$, we have
    \begin{align}
        \Err^\star\left(p \rho^{\otimes n}, \bar{p} \sigma^{\otimes n} \right)
        \asymp \frac{p^{\alpha^{\star}}\bar{p}^{1-\alpha^\star}}{\alpha^{\star}(1-\alpha^{\star})\sqrt{2\pi n V_{\alpha^\star}(\rho\Vert\sigma)}} \e^{ -n C(\rho, \sigma) },
    \end{align}
    where the Chernoff distance is defined as $C(\rho, \sigma) \coloneqq \sup_{\alpha \in (0,1) } - \log \Tr\left[\rho^{\alpha} \sigma^{1-\alpha}\right]$, and $V_{\alpha^\star}(\rho\Vert\sigma)$ is defined in \eqref{eq:V}.
    Here, the asymptotic equivalence $\asymp$ is up to the factor $4$.
\end{theorem}
\begin{remark}
We call Theorem~\ref{theorem:asymptotics_symmetric} an almost-exact asymptotic result because the ratio between the upper and lower bounds is only $4 + o(1)$.
\end{remark}

\begin{proof}
    The one-shot bound \eqref{eq:same_order} (resulting from Theorem~\ref{theorem:harmonic-mean_bound} applied to the i.i.d.~setting and the trace-preserving and tensor factorization properties of the \NS distributions (Fact~\ref{fact:NS}) directly show that
    \begin{align}
        \frac12 \cdot \Err^\star\left(p P^{\otimes n},\bar{p} Q^{\otimes n}\right)
    \leq
    \Err^\star\left(p \rho^{\otimes n}, \bar{p} \sigma^{\otimes n}\right)
    \leq 2 \cdot \Err^\star\left(p P^{\otimes n},\bar{p} Q^{\otimes n}\right).
    \end{align}
    We then apply the strong large-deviation theory (see e.g.~\cite{BR60}, \cite{Petrov1965}, \cite[Chapters VII, VIII]{Petrov1975},  \cite[\S 2.2]{DZ98}) on $\Err^\star\left(p P^{\otimes n},\bar{p} Q^{\otimes n}\right)$ and noting that $C(P,Q) = C(\rho,\sigma)$ (Fact~\ref{fact:NS}) completes the proof.
\end{proof}

Applying a similar idea to multiple hypothesis testing with the pairwise bound established in Theorem~\ref{theorem:pairwise}, we obtain the almost-exact asymptotics, wherein the worst pairwise Chernoff distance will dominate the overall performance.

\begin{theorem}[Almost-exact asymptotics for multiple hypothesis testing] \label{theorem:asymptotics_multiple}
Consider a quantum ensemble $(p_i, \rho_i)_{i=1}^M$ and define the dominant index set
\begin{align}
    \mathcal{I}^{\star}
    \coloneqq \left\{ (i,j) : C(\rho_i,\rho_j) = C(\rho_1, \ldots, \rho_M) \right\}.
\end{align}
Provided that the \NS distributions of each pair $(\rho_i, \rho_j)$, $(i,j)\in\mathcal{I}^{\star}$, satisfy~\eqref{eq:assumption_symmetric}, we have
\begin{align}
        \Err^\star\left(p_1 \rho_1^{\otimes n}, \cdots, p_M \rho_M^{\otimes n} \right)
        \asymp \sum_{(i,j)\in\mathcal{I}^{\star}}\frac{p_i^{\alpha_{ij}^\star}p_j^{1-\alpha_{ij}^\star}}{\alpha_{ij}^{\star}(1-\alpha_{ij}^{\star})\sqrt{2\pi n V_{\alpha_{ij}^\star}(\rho_i\Vert\rho_j)}} \e^{ -n C(\rho_1, \ldots, \rho_M) },
    \end{align}
    where each $\alpha_{ij}^\star \in (0,1)$ satisfies
    \begin{align}
        C(\rho_i,\rho_j) = C(\rho_1, \ldots, \rho_M) 
         = (1-\alpha_{ij}^{\star}) D_{\alpha_{ij}^{\star}}(\rho_i \Vert \rho_j).
    \end{align}
    Here, the asymptotic equivalence $\asymp$ is up to the factor $8\cdot|\mathcal{I}^{\star}|$.
\end{theorem}
\begin{proof}
The lower bound essentially follows from the coarse-graining reduction in \cite{NussbaumSzkola2011pure, Li16, sample_complexity_25}:
\begin{align}
    \Err^{\star}(p_1 \rho_1, \ldots, p_M \rho_{M})
    &\geq p_i \Tr[\rho_i (\I - \Lambda_i)] + p_j \Tr[\rho_j (\I - \Lambda_j)]
    \\
    &\geq p_i \Tr[\rho_i (\I - \Lambda_i)] + p_j \Tr[\rho_j \Lambda_i],
    &&\because \I - \Lambda_j = \Lambda_i + \sum_{k \neq i,j} \Lambda_k,
\end{align}
for any POVM $\{\Lambda_i\}_{i=1}^M$.
Hence,
\begin{align}
    \Err^{\star}(p_1 \rho_1, \ldots, p_M \rho_{M})
    \geq \max_{(i,j):\, i < j} \Err^\star(p_i \rho_i, p_j \rho_j).
\end{align}
Then, Theorem~\ref{theorem:asymptotics_symmetric} provides the lower bound:
\begin{align}
    \Err^\star\left(p_1 \rho_1^{\otimes n}, \cdots, p_M \rho_M^{\otimes n} \right) \geq\max_{(i,j)\in\mathcal{I}^{\star}}\frac{p_i^{\alpha_{ij}^\star}p_j^{1-\alpha_{ij}^\star}(1-o(1))}{2\alpha_{ij}^{\star}(1-\alpha_{ij}^{\star})\sqrt{2\pi n V_{\alpha_{ij}^\star}(\rho_i\Vert\rho_j)}} \e^{ -n C(\rho_1, \ldots, \rho_M)},
\end{align}
because the pairs with $C(\rho_{\tilde{i}}, \rho_{\tilde{j}}) > C(\rho_1, \ldots, \rho_M)$ decay exponentially faster and vanish relative to the dominant pairs in $\mathcal{I}^{\star}$.

For achievability, the pairwise bound in Theorem~\ref{theorem:pairwise} implies that
\begin{align}
    \Err^{\star}\left(p_1 \rho_1, \ldots, p_M \rho_M\right)
\leq 4 \sum_{(i,j):\,i<j} \Err^\star \!\big(p_i P^{ij}, p_j Q^{ij}\big),
\end{align}
where $(P^{ij}, Q^{ij})$ are the \NS distributions for $(\rho_i,\rho_j)$.
The tensor factorization property (Fact~\ref{fact:NS}) and the i.i.d.~asymptotics of $\big(P^{ij}\big)^{\otimes n}$ and $\big(Q^{ij}\big)^{\otimes n}$ yield
\begin{align}
    \Err^\star\left(p_1 \rho_1^{\otimes n}, \cdots, p_M \rho_M^{\otimes n} \right)
    \leq\sum_{(i,j)\in\mathcal{I}^{\star}}\frac{4\cdot p_i^{\alpha_{ij}^\star}p_j^{1-\alpha_{ij}^\star}(1+o(1))}{\alpha_{ij}^{\star}(1-\alpha_{ij}^{\star})\sqrt{2\pi n V_{\alpha_{ij}^\star}(\rho_i\Vert\rho_j)}} \e^{ -n C(\rho_1, \ldots, \rho_M) },
\end{align}
concluding the proof.
\end{proof}

\section{Asymmetric Quantum Hypothesis Testing} \label{sec:asymmetric}
In this section, we consider asymmetric hypothesis testing, where the focus is the optimal trade-off between the type-I and type-II error.
Define the minimum type-I error for discriminating $\rho$ against $\sigma$ when the type-II error is below some threshold $\mu$:
\begin{align} \label{eq:defn:alpha_hat}
    \Err^{\text{(I)}}_{\mu}(\rho\Vert\sigma)
    \coloneqq
    \inf_{0\leq T \leq \I} \left\{  \Tr[\rho (\I-T) ] : \Tr[\sigma T] \leq \mu \right\}.
\end{align}
Similarly, we define the minimum type-II error:
\begin{align} \label{eq:defn:beta_hat}
    \Err^{\text{(II)}}_{\varepsilon}(\rho\Vert\sigma)
    \coloneqq
    \inf_{0\leq T \leq \I} \left\{ \Tr[\sigma T] : \Tr[\rho (\I-T) ] \leq \varepsilon \right\}.
\end{align}

\begin{proposition}[Variational expression {\cite[Proposition 3.2]{AudenaertMosonyiFrank2012}}] \label{prop:variational_type-I}
For all density operators $\rho$ and $\sigma$ and all parameters $\mu, \varepsilon\geq 0$, we have 
    \begin{align}
        \Err^{\textnormal{(I)}}_{\mu}(\rho\Vert\sigma)
        &= \sup_{\gamma \geq 0} \left\{ \Tr\left[ \rho \wedge \gamma \sigma \right] - \gamma \mu \right\},
        \\
        \Err^{\textnormal{(II)}}_{\varepsilon}(\rho\Vert\sigma)
        &= \sup_{\gamma \geq 0} \left\{ \Tr\left[ \gamma \rho \wedge \sigma \right] - \gamma \varepsilon \right\},
    \end{align}
\end{proposition}
\begin{remark}
Proposition~\ref{prop:variational_type-I} may be equivalently viewed as the ``type-I-error version'' of the expressions obtained in \cite[Theorem 4]{regula2026tight}.
\end{remark}
\begin{proof}
    For $\mu > 0$, we apply the Lagrange method with Sion's minimax identity for linear functionals:\footnote{Banach-Alaoglu theorem guarantees that the set $\{0\leq T \leq \I\}$ is weak-* compact. Furthermore, for any trace-class operator $\tau$, the functional $T \mapsto \Tr[\tau T]$ is normal (weak-* continuous), which shows that our objective function is weak-* continuous and affine in $T$.}
    \begin{align}
        \Err^{\text{(I)}}_{\mu}(\rho\Vert\sigma)
    &=
    \inf_{0\leq T \leq \I} \left\{  \Tr[\rho (\I-T) ] : \Tr[\sigma T] \leq \mu \right\}
    \\
    &=  \inf_{0\leq T \leq \I} \sup_{\gamma\geq 0} \Tr[\rho (\I-T) ] +  \gamma (\Tr[\sigma T] - \mu) 
    \\
    &= \sup_{\gamma\geq 0} \inf_{0\leq T \leq \I} \Tr\left[ \rho - T (\rho - \gamma\sigma) \right] - \gamma \mu 
    \\
    &= \sup_{\gamma\geq 0} \Tr\left[ \rho - (\rho - \gamma\sigma)_+ \right] - \gamma \mu  
    \\
    &= \sup_{\gamma\geq 0} \Tr\left[ \rho \wedge \gamma\sigma\right] - \gamma \mu.
    \end{align}
    For $\mu = 0$, both representations evaluate $\Tr[\rho \sigma^0]$ by the direct-sum decomposition of $\sigma$ to its support and kernel.

    Switching the roles of $\rho$ and $\sigma$, we obtain the variational expression for $\Err^{\textnormal{(II)}}_{\varepsilon}(\rho\Vert\sigma)$.
\end{proof}

\begin{proposition}[One-shot relations] \label{prop:one-shot_type-I}
For all density operators $\rho$ and $\sigma$ with the corresponding \NS distributions $(P,Q)$ given in Definition~\ref{defn:NS} and all parameter $\mu\geq 0$,
we have, for all $s\in(0,1)$,
\begin{align}
    &\text{(upper bound)} \qquad \Err^{\textnormal{(I)}}_{\mu}(\rho\Vert\sigma)
    \leq \sup_{\gamma \geq 0} \left\{ \HM_s(\rho,\gamma\sigma) - \gamma \mu \right\}
    \leq \frac{1}{1-s} \Err^{\textnormal{(I)}}_{s\mu}(P\Vert Q),
    \\
    &\text{(lower bound)} \qquad 
    \Err^{\textnormal{(I)}}_{\mu}(\rho\Vert\sigma)
    \geq \sup_{\gamma \geq 0} \left\{ (1-s) \HM_s(\rho,\gamma\sigma) - \frac{1-s}{s}\gamma \mu \right\}
    \geq (1-s)\Err^{\textnormal{(I)}}_{\mu/s}(P\Vert Q).
\end{align}
\end{proposition}

\begin{proof}
    Using the variational formula in Proposition~\ref{prop:variational_type-I}, we derive the upper bound via the harmonic-mean bound established in Theorem~\ref{theorem:harmonic-mean_bound}:
    For all $s\in (0,1)$, we derive
    \begin{align}
        \sup_{\gamma \geq 0} \left\{ \Tr\left[ \rho \wedge \gamma \sigma \right] - \gamma \mu \right\}
        &\leq 
        \sup_{\gamma \geq 0} \left\{ \Tr\left[ P \mathbin{!}_s \gamma Q \right] - \gamma \mu \right\}
        &&\because \text{Theorem~\ref{theorem:harmonic-mean_bound}}
        \\
        &= \frac{1}{1-s} \sup_{\gamma' \geq 0} \left\{ \Tr\left[\frac{ P \cdot \gamma'Q}{P+\gamma' Q}\right] - s \gamma' \mu \right\}
        &&\because\gamma' = \frac{1-s}{s}\gamma
        \\
        &\leq \frac{1}{1-s} \sup_{\gamma' \geq 0} \left\{ \Tr\left[ P \wedge \gamma ' Q \right] - s \gamma' \mu \right\}
        \\
        &= \frac{1}{1-s} \Err^{\text{(I)}}_{s\mu}(P\Vert Q).
    \end{align}

    For the lower bound, we deduce
    \begin{align}
        \sup_{\gamma \geq 0} \left\{ \Tr\left[ \rho \wedge \gamma \sigma \right] - \gamma \mu \right\}
        &\geq \sup_{\gamma \geq 0} \left\{ \Tr\left[\frac{ P \cdot \gamma Q}{P+\gamma Q}\right] -  \gamma \mu \right\}
        &&\because \text{Theorem~\ref{theorem:harmonic-mean_bound}}
        \\
        &= \sup_{\gamma' \geq 0} \left\{ (1-s) \Tr\left[ P \mathbin{!}_s \gamma' Q \right] - \frac{1-s}{s} \gamma ' \mu \right\}
        &&\because\gamma' = \frac{s}{1-s}\gamma
        \\
        &\geq \sup_{\gamma' \geq 0} \left\{ (1-s) \Tr\left[ P \wedge \gamma' Q \right] - \frac{1-s}{s} \gamma ' \mu \right\}
        \\
        &= (1-s) \Err^{\text{(I)}}_{\mu/s}(P\Vert Q).
    \end{align}
\end{proof}


\medskip
As in Section~\ref{sec:asymptotics}, we derive the almost-exact asymptotics for asymmetric quantum hypothesis testing.
Define an error exponent function:
\begin{align}
E(R)
\coloneqq 
\sup_{\alpha\in(0,1)} \frac{1-\alpha}{\alpha} \left( D_{\alpha}(\rho\Vert\sigma) - R \right).
\end{align}

We assume the standard conditions that the non-orthogonal states $\rho$ and $\sigma$ are not identical on the common support, and the associated \NS distributions $(P,Q)$ are non-lattice.

\begin{fact}[Cram\'er's condition, asymmetric setting {\cite{BR60, Hoe65}, \cite[\S 2.2]{DZ98}}] \label{fact:condition_asymmetric}
Suppose the pair of density operators $(\rho,\sigma)$ and the corresponding \NS distributions $(P,Q)$ satisfy
\begin{subequations} \label{eq:assumption_asymmetric}
\begin{align}
    &D_0(\rho\Vert\sigma) < D(\rho\Vert\sigma), \label{eq:assumption_asymmetric1}
    \\
    &\text{$(P,Q)$ are non-lattice.}
    \label{eq:assumption_asymmetric2}
\end{align}
\end{subequations}
Then, there exists a unique $\alpha^{\star} \in (0,1)$ satisfying
\begin{align}
&E(R)
= \frac{1-\alpha^\star}{\alpha^\star} \left( D_{\alpha^\star}(\rho\Vert\sigma) - R \right) = D\left( P_{\alpha^\star} \Vert P \right),
\label{eq:E(R)_Cramer}
\\
&R = D\left( P_{\alpha^\star} \Vert Q \right),
\\
&V_{\alpha^\star}(\rho\Vert\sigma)
\coloneqq \textnormal{Var}_{P_{\alpha^\star}}\left[ \log \frac{P}{Q} \right] \in (0,+\infty),
\end{align}
where the tilted distribution is defined as
\begin{align}
    P_{\alpha} \coloneqq \frac{P^{\alpha} Q^{1-\alpha}}{\Tr\left[P^{\alpha} Q^{1-\alpha}\right]}.
\end{align}
\end{fact}

\begin{theorem}[Almost-exact asymptotics for asymmetric hypothesis testing] \label{theorem:asymptotic_asymmetric}
For all density operators $\rho$ and $\sigma$ and the corresponding \NS distributions $(P,Q)$
satisfying ~\eqref{eq:assumption_asymmetric},
we have, for any rate $R$ satisfying $D_0(\rho\Vert\sigma) < R < D(\rho\Vert\sigma)$,
\begin{align}
    2^{-\frac{1}{\alpha^\star}} (1-o(1)) \Err^{\textnormal{(I)}}_{\e^{-nR}}\left(P^{\otimes n}\Vert Q^{\otimes n} \right)
    \leq
    \Err^{\textnormal{(I)}}_{\e^{-nR}}\left(\rho^{\otimes n}\Vert \sigma^{\otimes n} \right)
    \leq 2^{\frac{1}{\alpha^\star}} (1+o(1)) \Err^{\textnormal{(I)}}_{\e^{-nR}}\left(P^{\otimes n}\Vert Q^{\otimes n} \right),
\end{align}
where $\alpha^{\star} \in (0,1)$ is a unique parameter for \eqref{eq:E(R)_Cramer}.

In particular, 
\begin{align}
    \Err^{\textnormal{(I)}}_{\e^{-nR}}\left(\rho^{\otimes n}\Vert \sigma^{\otimes n} \right)
    &\asymp 
     \frac{1}{1-\alpha^\star} \left( \frac{1}{\alpha^\star} \right)^{\frac{1-\alpha^\star}{\alpha^\star}} \left(2\pi n V_{\alpha^\star}(\rho\Vert\sigma) \right)^{-\frac{1}{2\alpha^\star}}
     \e^{-n E(R)},
\end{align}
where $E(R)$ is given in \eqref{eq:E(R)_Cramer}.
Here, the asymptotic equivalence $\asymp$ is up to the factor $2^{\frac{2}{\alpha^{\star}}}$.
\end{theorem}
\begin{proof}
    Applying Proposition~\ref{prop:one-shot_type-I} with $s = \frac12$, we have
    \begin{align}
        \frac12 \cdot \Err^{\textnormal{(I)}}_{\e^{-nR+\log 2}}\left(P^{\otimes n}\Vert Q^{\otimes n} \right)
        \leq \Err^{\textnormal{(I)}}_{\e^{-nR}}\left(\rho^{\otimes n}\Vert \sigma^{\otimes n} \right)
        \leq 2 \cdot \Err^{\textnormal{(I)}}_{\e^{-nR- \log 2}}\left(P^{\otimes n}\Vert Q^{\otimes n} \right).
    \end{align}
    Then, we employ Fact~\ref{fact:condition_asymmetric} and the first-order condition to calculate 
    \begin{align}
    \frac{\d E(R)}{\d R} &=\left. \frac{\partial}{\partial \alpha} \left[ \frac{1-\alpha}{\alpha} \big( D_{\alpha}(\rho\|\sigma) - R \big) \right] \right|_{\alpha=\alpha^\star(R)} \cdot \frac{\mathrm{d}\alpha^\star(R)}{\mathrm{d}R} \quad + \left. \frac{\partial}{\partial R} \left[ \frac{1-\alpha}{\alpha} \big( D_{\alpha}(\rho\|\sigma) - R \big) \right] \right|_{\alpha=\alpha^\star(R)}
    \\
    &= 0 \cdot \frac{\mathrm{d}\alpha^\star(R)}{\mathrm{d}R} + \frac{\partial}{\partial R} \left[ \frac{1-\alpha^\star(R)}{\alpha^\star(R)} D_{\alpha^\star(R)}(\rho\|\sigma) - \frac{1-\alpha^\star(R)}{\alpha^\star(R)} R \right]
    \\
    &= -\frac{1-\alpha^\star}{\alpha^\star}.
    \end{align}
    We apply the Taylor series expansion of $E(\cdot)$ around $R$ to remove the rate deviation terms $\pm \frac{1}{n} \log 2$, which incurs a factor $2^{\pm \frac{1-\alpha^{\star}}{\alpha^{\star}}}(1\pm o(1))$.
    Combining with the prefactor $2^{\pm 1}$ completes the proof.
\end{proof}

\begin{remark} \label{eq:condition_asymmetric}
    We only focus on $D_0(\rho\Vert\sigma) < R < D(\rho\Vert\sigma)$, because, otherwise, we would have
    \begin{align}
        E(R) 
        = 
        \begin{dcases}
            +\infty & R\leq D_0(\rho\Vert\sigma),
            \\
            0 & R \geq D(\rho\Vert\sigma).
        \end{dcases}
    \end{align}
\end{remark}

\medskip
For a pair of density operators $\rho$ and $\sigma$, we define the quantum hypothesis testing divergence as
\begin{align}
    D_{\text{H}}^{\varepsilon}(\rho\Vert\sigma)
    \coloneqq - \log \Err^{\text{(II)}}_{\varepsilon}(\rho\Vert\sigma).
\end{align}

\begin{proposition}[Improved third-order converse] \label{prop:third-order_converse}
    For all density operators $\rho$ and $\sigma$ satisfying $\supp(\rho)\subseteq\supp(\sigma)$ and $\rho \neq \sigma$, and assuming their  Nussbaum--Szko{\l}a distributions $(P,Q)$ satisfy $\Tr\big[P\,\left|\log P - \log Q\right|^3\big] < \infty$,
    we have
    \begin{align}
    D_{\textnormal{H}}^{\varepsilon}\left(\rho^{\otimes n}\Vert\sigma^{\otimes n}\right)
        \leq n D(\rho\|\sigma) + \sqrt{nV(\rho\Vert\sigma)} \Phi^{-1}(\varepsilon) + \frac12 \log n + \mathcal{O}(1), \quad \forall\,\varepsilon \in (0,1),
    \end{align}
    where ${\Phi}^{-1}(\eps) \coloneqq \sup\{u: \int_{-\infty}^u \frac{1}{\sqrt{2\pi}} \e^{-\frac12 t^2}\mathrm{d}t \leq \eps \}$ is the inverse of the cumulative distribution function of the standard normal distribution, and $V(\rho\Vert\sigma) \coloneqq \Tr[\rho(\log \rho - \log \sigma)^2] - D(\rho\Vert\sigma)^2$ is the relative entropy variance \cite{TH13, Li14}.
\end{proposition}

\begin{remark}
    Proposition~\ref{prop:third-order_converse} improves the third-order term in \cite[Corollary 1]{LVK26} from $+\log n$ to $+\frac12 \log n$.
    The gain lies in the asymptotic expansion on the harmonic mean directly, instead of on the end-point quantity $s \Err^{\text{(II)}}_{\varepsilon/(1-s)}(P\Vert Q)$, $s\in (0,1)$.
    Applying a similar idea to achievability, we recover Li's third-order term in \cite{Li14}.
\end{remark}

\begin{proof}
    As in the proof of Proposition~\ref{prop:one-shot_type-I}, the variational formula in Proposition~\ref{prop:variational_type-I} with Theorem~\ref{theorem:harmonic-mean_bound} show that
    \begin{align}
        \Err^{\text{(II)}}_{\varepsilon}(\rho\Vert\sigma)
        \geq \sup_{\gamma \geq 0} \left\{ \Tr\left[\frac{ \gamma P^{\otimes n} \cdot Q^{\otimes n}}{\gamma P^{\otimes n } + Q^{\otimes n}}\right] -  \gamma \varepsilon \right\}
        =: \Gamma_{\varepsilon}^{(n)}.
    \end{align}
    In the following, we establish the classical asymptotic expansion
    \begin{align}
    - \log \Gamma_{\varepsilon}^{(n)}
    = n D(\rho\|\sigma) + \sqrt{n V(\rho\|\sigma)} \Phi^{-1}(\varepsilon) + \frac12 \log n + \mathcal{O}(1),
    \end{align}
    which completes the proof.

    Let $Z_n \coloneqq \log \frac{P^{\otimes n}}{Q^{\otimes n}}$.
    We rewrite $\Gamma_{\varepsilon}^{(n)}$ as an expectation under $P^{\otimes n}$:
    \begin{equation}
        \Gamma_{\varepsilon}^{(n)} 
        = \sup_{\gamma \geq 0} \gamma \left( \mathbf{E}_{P^{\otimes n}} \left[ \frac{1}{1 + \gamma \e^{Z_n}} \right] - \varepsilon \right).
    \end{equation}
    Let $\Upsilon_n(\gamma) \coloneqq \mathbf{E}_{P^{\otimes n}} \left[ (1 + \gamma \e^{Z_n})^{-1} \right]$.
    The objective function is $f_n(\gamma) = \gamma (\Upsilon_n(\gamma) - \varepsilon)$. 
    
    Because $\rho \neq \sigma$ and $\supp(\rho) \subseteq \supp(\sigma)$, the variance satisfies 
    \begin{align}
        V(\rho\|\sigma)
        = \Var_{P} \left[ \log\frac{P}{Q} \right] > 0.
    \end{align}
    We find the optimal threshold $\gamma_n^{\star} \geq 0$ via the derivative test:
    \begin{equation} \label{eq:derivative}
        f_n'(\gamma_n^{\star}) = \Upsilon_n(\gamma_n^{\star}) - \varepsilon + \gamma_n^{\star} \Upsilon_n'(\gamma_n^{\star}) = 0.
    \end{equation}
    Define the remainder kernel function $K_n(\gamma) \coloneqq - \gamma \Upsilon_n'(\gamma)$.
    Differentiating the expectation inside the integral yields:
    \begin{equation}
        K_n(\gamma) = \mathbf{E}_{P^{\otimes n}} \left[ \frac{\gamma \e^{Z_n}}{(1 + \gamma \e^{Z_n})^2} \right].
    \end{equation}
    From \eqref{eq:derivative} and substituting this back into the objective function $f_n$, we obtain
    \begin{align} \label{eq:max_value}
        &\Gamma_{\varepsilon}^{(n)} = f_n(\gamma_n^{\star}) = \gamma_n^{\star} K_n(\gamma_n^{\star}),
        \\
        &- \log \Gamma_{\varepsilon}^{(n)} = - \log \gamma_n^{\star} - \log K_n(\gamma_n^{\star}).
        \label{eq:log_decomp}
    \end{align}

    Now let $X_n \coloneqq Z_n + \log \gamma_n^{\star}$. The term $K_n(\gamma_n^{\star})$ can be expressed as the expectation of the standard logistic density function $k(x) = \frac{\e^x}{(1+\e^x)^2}$:
    \begin{equation}
        K_n(\gamma_n^{\star}) = \mathbf{E}_{P^{\otimes n}} [k(X_n)] = \int_{-\infty}^{\infty} k(x) \mathbf{P}_{X_n}(x) \mathrm{d}x.
    \end{equation}
    We then evaluate this integral using the local limit theorem \cite[Chapter VII]{Petrov1975}, \cite[Theorem 2.3]{Chaganty1993}.
    Under the hypothesis that the third absolute moment is finite and the variance $V(\rho\|\sigma)$ is strictly positive, the local limit theorem guarantees that the probability density of $X_n$ near the centered threshold satisfies $\mathbf{P}_{X_n}(x) = \frac{1}{\sqrt{2\pi n V(\rho\|\sigma)}}(1 + o(1))$.
    (If the underlying distribution is lattice, an identical asymptotic scaling holds for the probability mass function evaluated at the lattice points). 
    
    Because the kernel $k(x)$ decays exponentially fast away from the origin and integrates to unity, 
    passing the limit through the integral yields:
    \begin{equation}
        K_n(\gamma_n^{\star}) = \frac{1}{\sqrt{2\pi n V(\rho\|\sigma)}}(1 + o(1)) \int_{-\infty}^{\infty} k(x) \mathrm{d}x = \Theta(n^{-1/2}).
    \end{equation}
    Taking the negative logarithm and absorbing the variance, higher-order limits, and any lattice constants into the $\mathcal{O}(1)$ remainder, we have
    \begin{equation} \label{eq:Kn_eval}
        - \log K_n(\gamma_n^{\star}) = \frac{1}{2} \log n + \mathcal{O}(1).
    \end{equation}

    Lastly, from the critical point condition, we know 
    \begin{equation}
        \Upsilon_n(\gamma_n^{\star}) = \varepsilon + K_n(\gamma_n^{\star}) = \varepsilon + \mathcal{O}(n^{-\nicefrac{1}{2}}).    
    \end{equation}
    To connect $\Upsilon_n(\gamma_n^{\star})$ to the cumulative distribution function of $X_n$, let $F_{X_n}(x) \coloneqq \mathbf{P}_{P^{\otimes n}}(X_n \leq x)$. We express the expectation as a Lebesgue--Stieltjes integral and apply integration by parts:
    \begin{align}
        \Upsilon_n(\gamma_n^{\star}) 
        = \int_{-\infty}^{\infty} \frac{1}{1+\e^x} \mathrm{d}F_{X_n}(x) 
        = \left[ \frac{F_{X_n}(x)}{1+\e^x} \right]_{-\infty}^{\infty} - \int_{-\infty}^{\infty} F_{X_n}(x) \frac{\mathrm{d}}{\mathrm{d}x}\left(\frac{1}{1+\e^x}\right) \mathrm{d}x.
    \end{align}
    The boundary terms evaluate to zero. Since the derivative is precisely the negative of the standard logistic density $k(x)$, we obtain:
    \begin{equation}
        \Upsilon_n(\gamma_n^{\star}) = \int_{-\infty}^{\infty} F_{X_n}(x) k(x) \mathrm{d}x.
    \end{equation}
    We compare this to the sharp threshold probability $F_{X_n}(0) = \mathbf{P}_{P^{\otimes n}}(Z_n \leq -\log \gamma_n^{\star})$. The difference is:
    \begin{equation} \label{eq:integral_diff}
        \Upsilon_n(\gamma_n^{\star}) - F_{X_n}(0) = \int_{-\infty}^{\infty} (F_{X_n}(x) - F_{X_n}(0)) k(x) \mathrm{d}x.
    \end{equation}
    Because we assumed a finite third moment, the Berry--Esseen theorem \cite{Berry1941, Esseen1945} dictates that $F_{X_n}(x)$ is uniformly approximated by a Gaussian cumulative distribution function $G_n(x)$ with an error bounded by $\mathcal{O}(n^{-\nicefrac{1}{2}})$. Furthermore, because the Gaussian density scales as $\mathcal{O}(n^{-\nicefrac{1}{2}})$, the Mean Value Theorem implies $|G_n(x) - G_n(0)| \leq \mathcal{O}(|x| n^{-\nicefrac{1}{2}})$. Therefore, we can bound the difference pointwise:
    \begin{equation}
        |F_{X_n}(x) - F_{X_n}(0)| \leq \mathcal{O}(n^{-\nicefrac{1}{2}}) + \mathcal{O}(|x| n^{-\nicefrac{1}{2}}) = \mathcal{O}((1+|x|) n^{-\nicefrac{1}{2}}).
    \end{equation}
    Substituting this bound into \eqref{eq:integral_diff} yields
    \begin{equation}
        |\Upsilon_n(\gamma_n^{\star}) - F_{X_n}(0)| \leq \mathcal{O}(n^{-\nicefrac{1}{2}}) \int_{-\infty}^{\infty} (1+|x|) k(x) \mathrm{d}x = \mathcal{O}(n^{-\nicefrac{1}{2}}),
    \end{equation}
    where the final equality holds because the logistic distribution has finite absolute first moments. Thus, we have shown:
    \begin{equation}
        \mathbf{P}_{P^{\otimes n}}(Z_n \leq -\log \gamma_n^{\star}) = \varepsilon + \mathcal{O}(n^{-\nicefrac{1}{2}}).
    \end{equation}
    By the standard Edgeworth expansion for $Z_n$, we obtain
    \begin{equation}
        \Phi\left( \frac{-\log \gamma_n^{\star} - n D(P\|Q)}{\sqrt{n V(P\|Q)}} \right) = \varepsilon + \mathcal{O}(n^{-\nicefrac{1}{2}})
    \end{equation}
    (by the assumption of the finite third absolute moment of the log-likelihood ratio $Z_1$).
    This leads us to
    \begin{equation} \label{eq:gamma_eval}
        -\log \gamma_n^{\star} = n D(P\|Q) + \sqrt{n V(P\|Q)} \Phi^{-1}(\varepsilon) + \mathcal{O}(1).
    \end{equation}
    Combining \eqref{eq:log_decomp}, \eqref{eq:Kn_eval}, and \eqref{eq:gamma_eval}, and recalling that $D(P\|Q) = D(\rho\|\sigma)$ and $V(P\|Q) = V(\rho\|\sigma)$ (Fact~\ref{fact:NS}), we complete the proof.
\end{proof}

\section*{Acknowledgments}
We sincerely thank Bartosz Regula and Ludovico Lami for their valuable discussions and suggestions.
We thank AIs for polishing the manuscript.
HC would like to thank Bartosz Regula for his hospitality during the visit in RIKEN, Japan, and would like to thank Ludovico Lami and Filippo Girardi for hosting the wonderful workshop: \emph{\href{https://ypetz.filippo.info/}{YPetz --- Advances in quantum hypothesis testing and quantum entropies}} at Scuola Normale Superiore, Pisa, Italy.
We are supported by the Emerging Young Scholars Program of the National Science and Technology Council, Taiwan (R.O.C.) under Grant numbers~NSTC 114-2628-E-002-006, NSTC 114-2119-M-001-002, and NSTC 114-2124-M-002-003, by the Yushan Young Scholar Program of the Ministry of Education, Taiwan (R.O.C.) under Grant number~NTU-114V2016-1, and by the research project `Forefront Quantum Computing, Learning, and Engineering in Noisy Intermediate-Scale Quantum Era’ of National Taiwan University under Grant NTU-114L895005. H.-C.~Cheng acknowledges the support from the `Center for Advanced Computing and Imaging in Biomedicine (NTU-115L900702)' through The Featured Areas Research Center Program within the framework of the Higher Education Sprout Project by the Ministry of Education (MOE) in Taiwan.

\appendix

\section{Quantum Union Bound via The Data-Processing Inequality} \label{sec:union_bound}

Let $\rho$ be a density operator on a (separable) Hilbert space $\cH$, and let $\Pi_1,\dots,\Pi_m$ be orthogonal projections on $\cH$. Write $\bar\Pi_t\coloneqq\I-\Pi_t$. Define the Kraus operators
\begin{align}
K_0:=\Pi_m\Pi_{m-1}\cdots\Pi_1,
\qquad
K_t:=\bar\Pi_t\Pi_{t-1}\cdots\Pi_1
\quad (1\leq t\leq m).
\end{align}
We also define
\begin{align}
\Succ \coloneqq \Tr\left[ K_0\rho K_0^{\dagger} \right],
\qquad
\Fail \coloneqq 1-\Succ,
\qquad
\Loss \coloneqq \sum_{t=1}^m \Tr\left[ \rho\bar\Pi_t \right],
\end{align}
and, when $\Succ>0$,
\begin{align}
\rho_m:=\frac{K_0\rho K_0^{\dagger}}{\Succ}.
\end{align}

The \emph{quantum union bound} aims to upper bound the probability of {\Fail} (i.e.~the probability of not all $\Pi_t$ occurring) by the total \Loss. 
A sequence of works by Sen~\cite{Sen12}, Gao~\cite{Gao15}, Oskouei--Mancini--Wilde \cite{OMW19}, and O'Donnell--Venkateswaran~\cite{OV22} shows that {\Fail} can be upper bounded by {\Loss} with an only additional universal multiplicative factor.
Precisely, Gao's quantum union bound \cite{Gao15} states that
\begin{align}\label{eq:gao}
\textnormal{\Fail}\leq 4\textnormal{\Loss}.
\end{align}

For the purpose of this work, we re-write Gao's quantum union bound in \eqref{eq:gao} for a non-normalized state.
\begin{theorem}[Quantum union bound for non-normalized states) {\cite{Gao15, OMW19, OV22}}] \label{theorem:Union_Bound}
Let \(\rho\) be a positive semidefinite trace-class operator and let \(\Pi_1,\dots,\Pi_m\) be projections
(not assumed pairwise orthogonal). Then
\begin{align}
\Tr\left[ \Pi_m\cdots \Pi_1\,\rho\,\Pi_1\cdots \Pi_m\right]
\geq
\Tr[\rho] - 4\sum_{t=1}^m \Tr\left[(\I-\Pi_t)\rho\right].
\end{align}
\end{theorem}
\begin{proof}
If \(\Tr[\rho]=0\), there is nothing to prove. 
Otherwise apply Gao's quantum union bound in \eqref{eq:gao} to the normalized state \(\frac{\rho}{\Tr[\rho]}\), and then multiply both sides by
\(\Tr[\rho]\).
\end{proof}

Later, O'Donnell and Venkateswaran provided a surprisingly elegant and simple proof of Gao's quantum union bound in \eqref{eq:gao}.
Specifically, the O'Donnell--Venkateswaran inequality \cite[Theorem 1.3]{OV22} states that
\begin{align} \label{eq:OV}
1
\leq
\sqrt{\Succ}\,{F(\rho_m,\rho)}
+\sqrt{\Fail}\,\sqrt{\Loss},
\end{align}
where $F(\rho,\sigma) \coloneqq \|\sqrt{\rho}\sqrt{\sigma}\|_1$ is Uhlmann's fidelity.\footnote{Notice that O'Donnell and Venkateswaran adopted the square fidelity notation $\|\sqrt{\rho}\sqrt{\sigma}\|_1^2$ in \cite{OV22}.}
Provided that $\Loss \leq 1$, \eqref{eq:OV} then implies
\begin{align} \label{eq:OV2}
\Fail\leq \frac{4\Loss}{(1+\Loss)^2}\leq 4\Loss.
\end{align}

The aim here is to view the O'Donnell--Venkateswaran inequality through the lens of data-processing inequalities, offering a familiar framework for the quantum information community.

We begin by introducing necessary notation.
Let $A,B\geq 0$ be non-orthogonal positive semidefinite operators. For $0<\alpha<1$, define the measured R\'enyi divergence as
\begin{align}
D_{\alpha}^{\mathds M}(A\Vert B)
&\coloneqq
\frac{1}{\alpha-1}\log Q_{\alpha}^{\mathds M}(A\Vert B)
\\
Q_{\alpha}^{\mathds M}(A\|B)
&\coloneqq
\inf_{\{M_z\}_z }
\sum_z \bigl(\Tr[M_z A]\bigr)^\alpha \bigl(\Tr[M_z B]\bigr)^{1-\alpha},
\end{align}
where the infimum is over all POVMs on the underlying space. Lastly, define
\begin{align}
L_{\alpha}(\ell)
:=
\begin{cases}
1, & 0<\alpha\leq \tfrac12,\\
(1+\ell)^{1-2\alpha}, & \tfrac12\leq \alpha<1.
\end{cases}
\end{align}

\begin{proposition}[Measured R\'enyi O'Donnell--Venkateswaran inequality]\label{prop:measure_union_bound}
For every $0<\alpha<1$,
\begin{align}\label{eq:main-ineq}
L_{\alpha}(\textnormal{\Loss})
\leq
\textnormal{\Succ}^{\alpha}Q_{\alpha}^{\mathds M}(\rho_m\Vert\rho)
+\textnormal{\Fail}^{\alpha}\textnormal{\Loss}^{1-\alpha}.
\end{align}
Equivalently, for $\textnormal{\Succ}>0$,
\begin{align}\label{eq:main-divergence}
L_{\alpha}(\textnormal{\Loss})
\leq
(1-\textnormal{\Fail})^{\alpha}\e^{(\alpha-1)D_{\alpha}^{\mathds M}(\rho_m\|\rho)}
+ \textnormal{\Fail}^{\alpha}\textnormal{\Loss}^{1-\alpha}.
\end{align}
\end{proposition}
At $\alpha=\nicefrac12$, 
$Q_{1/2}^{\mathds M}(\rho_m\|\rho)={F(\rho_m,\rho)}$.
This becomes
\begin{align}
1
\leq
\sqrt{\textnormal{\Succ}}\,{F(\rho_m,\rho)}
+\sqrt{\textnormal{\Fail}}\,\sqrt{\textnormal{\Loss}},
\end{align}
which is exactly the O'Donnell--Venkateswaran inequality \eqref{eq:OV} \cite[Theorem 1.3]{OV22}.

\begin{proof}
Let $|\psi\rangle$ be a purification of $\rho$ with a reference system $\mathsf{R}$, and let $\mathsf{L}\cong\mathds{C}^{m+1}$ be a label register with orthonormal basis $\{|t\rangle_{\mathsf{L}}:0\leq t\leq m\}$.
Consider the vectors
\begin{align}
|\Phi\rangle &\coloneqq \sum_{t=0}^m |t\rangle_{\mathsf{L}}\otimes (K_t\otimes \I_\mathsf{R})|\psi\rangle,
\\
|\Omega\rangle &\coloneqq |0\rangle_{\mathsf{L}}\otimes |\psi\rangle + \sum_{t=1}^m |t\rangle_{\mathsf{L}}\otimes (\bar\Pi_t\otimes \I_{\mathsf{R}})|\psi\rangle.
\end{align}
Because the labels are orthogonal, we have
\begin{align}
\langle\Omega|\Phi\rangle
=
\langle\psi|\left( K_0+\sum_{t=1}^m \bar\Pi_t K_t \right)|\psi\rangle.
\end{align}
Since $\bar\Pi_t$ is a projection, $\bar\Pi_t K_t = K_t$ for every $t\geq 1$. Moreover,
\begin{align}\label{eq:sumK}
\sum_{t=0}^m K_t
=
\Pi_m\cdots\Pi_1+\sum_{t=1}^m \bar\Pi_t\Pi_{t-1}\cdots\Pi_1
=
\I,
\end{align}
obtained by repeatedly expanding $\I=\Pi_t+\bar\Pi_t$. Hence,
\begin{align}
\langle\Omega|\Phi\rangle=1.
\end{align}
On the other hand,
\begin{align}
\|\Phi\|_2^2
=
\sum_{t=0}^m \langle\psi|K_t^{\dagger}K_t|\psi\rangle
=
\langle\psi|\psi\rangle
=
1,
\end{align}
using $\sum_{t=0}^m K_t^{\dagger}K_t=\I$, and
\begin{align}
\|\Omega\|_2^2
=
1+\sum_{t=1}^m \langle\psi|\bar\Pi_t|\psi\rangle
=
1+\Loss.
\end{align}
Hence, a direct calculation shows that
\begin{align}\label{eq:left-pure}
Q_{\alpha}^{\mathds M}\bigl(|\Phi\rangle\!\langle\Phi|\,\|\,|\Omega\rangle\!\langle\Omega|\bigr)
&=
\begin{cases}
|\langle \Phi|\Omega\rangle|^{2(1-\alpha)}, & 0<\alpha\leq \tfrac12\\
\|\Omega\|_2^{2-4\alpha}\,|\langle \Phi|\Omega\rangle|^{2\alpha}, & \tfrac12\leq \alpha<1
\end{cases}
\\
&=
L_{\alpha}(\Loss).
\end{align}

Now let $\mathscr{D}_{\mathsf{L}}$ denote complete dephasing in the basis $\{|t\rangle_{\mathsf{L}}\}_{t=0}^m$, and define a channel
\begin{align}
\Lambda(X_0\oplus X_1\oplus\cdots\oplus X_m)
\coloneqq
X_0\oplus \left(\sum_{t=1}^m \Tr[X_t]\right)|1\rangle\!\langle 1|.
\end{align}
Moreover, let
\begin{align}
\Gamma\coloneqq \Lambda\circ(\mathscr{D}_{\mathsf{L}}\otimes \Tr_{\mathsf{R}}).
\end{align}
A direct calculation gives
\begin{align}
\Gamma\bigl(|\Phi\rangle\!\langle\Phi|\bigr)
&=
K_0\rho K_0^{\dagger}\oplus\left(\sum_{t=1}^m \Tr[K_t\rho K_t^{\dagger}]\right)|1\rangle\!\langle 1|\\
&=
\Succ\cdot\rho_m\oplus \Fail\,|1\rangle\!\langle 1|,
\end{align}
while similarly
\begin{align}
\Gamma\bigl(|\Omega\rangle\!\langle\Omega|\bigr)
&=
\rho\oplus\left(\sum_{t=1}^m \Tr[\bar\Pi_t\rho\bar\Pi_t]\right)|1\rangle\!\langle 1|\\
&=
\rho\oplus \Loss\,|1\rangle\!\langle 1|,
\end{align}
since $\bar\Pi_t^2=\bar\Pi_t$.

Applying data-processing inequality to $\Gamma$, and then using the direct-sum structure of $Q_{\alpha}^{\mathds M}$, we obtain
\begin{align}
L_{\alpha}(\Loss)
&=Q_{\alpha}^{\mathds M}\bigl(|\Phi\rangle\!\langle\Phi|\,\|\,|\Omega\rangle\!\langle\Omega|\bigr)\\
&\leq Q_{\alpha}^{\mathds M}\bigl(\Succ\cdot\rho_m\oplus \Fail\,|1\rangle\!\langle 1|\,\|\,\rho\oplus \Loss\,|1\rangle\!\langle 1|\bigr)\\
&=\Succ^{\alpha}Q_{\alpha}^{\mathds M}(\rho_m\|\rho)+\Fail^{\alpha}\Loss^{1-\alpha}Q_{\alpha}^{\mathds M}(|1\rangle\!\langle 1|\,\|\,|1\rangle\!\langle 1|)\\
&=\Succ^{\alpha}Q_{\alpha}^{\mathds M}(\rho_m\|\rho)+\Fail^{\alpha}\Loss^{1-\alpha},
\end{align}
which is \eqref{eq:main-ineq}. 
The divergence form \eqref{eq:main-divergence} is the identity
\begin{align}
Q_{\alpha}^{\mathds M}(\rho_m\|\rho) = \e^{(\alpha-1)D_{\alpha}^{\mathds M}(\rho_m\|\rho)}.
\end{align}
The proof is completed.
\end{proof}

\begin{remark}[Optimality]
    As for the quantum union bound, the factor $4$ in Gao's result, \eqref{eq:gao}, is already sharp, as noted in \cite[\S 3.6]{OV22}.
    Although the measured R\'enyi O'Donnell--Venkateswaran inequality (Proposition~\ref{prop:measure_union_bound}) could be tighter than the original O'Donnell--Venkateswaran inequality \eqref{eq:OV} \cite[Theorem 1.3]{OV22} for a certain range of parameters, the inequality of the form \eqref{eq:OV2} is not improved via Proposition~\ref{prop:measure_union_bound}.
\end{remark}

\section{The Nussbaum--Szkoła Lower Bound} \label{sec:NS}

For positive semidefinite trace-class operators $A,B\geq 0$ with the associated \NS distributions $(P,Q)$ introduced in Definition~\ref{defn:NS}, Nussbaum and Szko{\l}a \cite[(12)]{NS09}, \cite[Proposition 2]{ANS+08} proved the following lower bound to the minimum error $\Tr[A\wedge B]$, which plays an important role in various converse analyses of quantum information-theoretic tasks:
\begin{align} \label{eq:NS_lower}
	\frac12 \Tr[P \wedge Q] 
	\leq \Tr[A \wedge B].
\end{align}
The following result revisits the \NS lower bound in terms of the Petz--\NS harmonic mean given in Definition~\ref{defn:harmonic_mean}.
\begin{proposition}[\NS lower bound, revised] \label{prop:NS_lower}
	For all trace-class operators $A, B \geq 0$ with the \NS distributions $(P,Q)$ in \eqref{eq:defn:NS} and for all parameters $s \in (0,1)$, 
	\begin{align}
		&s\wedge (1-s) \cdot \Tr[P\wedge Q]
		\leq s\wedge (1-s) \cdot \Tr[P\mathbin{!}_s Q] 
		\leq \Tr[A \wedge B].
	\end{align}
\end{proposition}

\begin{proof}
	Recall the infimum representation \eqref{eq:infimum_representation} in Proposition~\ref{prop:representations} and the cyclic property of the trace:
	\begin{align}
		\Tr[P\mathbin{!}_s Q] = \inf_X \Tr\left[ \frac{1}{1-s} A (\I-X)(\I-X)^\dagger + \frac{1}{s} B X^\dagger X \right]. 
	\end{align}
	Choosing the Holevo--Helstrom test $X_0 := \left\{ A \geq B \right\}$ (i.e., the spectral projection of $A-B$ onto its nonnegative subspace) as the feasible solution and noting that $(\I-X_0)(\I-X_0)^\dagger = \I-X_0 \geq 0$ and $X_0^\dagger X_0 = X_0^2 = X_0 \geq 0$, we have
	\begin{align}
		s\wedge (1-s) \cdot \Tr[P\mathbin{!}_s Q]
		&\leq \frac{s\wedge (1-s)}{1-s} \Tr[A(\I-X_0)] + \frac{s\wedge (1-s)}{s} \Tr[B X_0]
		\\
		&\leq \Tr[A(\I-X_0)] + \Tr[B X_0]
		\\
		&= \Tr[A\wedge B],
	\end{align}
	proving the second inequality. 
	The first inequality directly follows from the scalar inequality $ x\wedge y\leq x\mathbin{!}_{s} y$, $x,y\geq 0$, $s \in (0,1)$.
\end{proof}

\begin{remark}[Historical notes]
    The endpoint inequality $s\wedge (1-s) \cdot \Tr[P\wedge Q] \leq \Tr[A \wedge B]$ with $s = \nicefrac12$ corresponds to the original \NS lower bound \eqref{eq:NS_lower}.
    In fact, Nussbaum and Szko{\l}a's original proof with a careful derivation (e.g., see \cite[p.~1047]{NS09}, \cite[(19)]{ANS+08}) could also yield the intermediate harmonic mean $s\wedge (1-s) \cdot \Tr[P\mathbin{!}_s Q]$.
    On the other hand, the proof of Proposition~\ref{prop:NS_lower} proceeds with another route from the infimum representation of the harmonic mean. 
    
    We also notice that \cite[Proposition 3.2]{JOP+12}, \cite[Lemma 2]{DPR16} actually derived the middle inequality $s\wedge (1-s) \cdot \Tr[P\mathbin{!}_s Q] \leq \Tr[A\wedge B]$ for the case $s = \nicefrac12$.
\end{remark}

\begin{remark}[Optimality]
The factor $\nicefrac12$ in \eqref{eq:NS_lower} is in general sharp.
For example, for $0 < \gamma \leq 1$, consider $\rho = \gamma |0\rangle \langle 0 | + (1-\gamma) |2\rangle \langle 2|$ and $\sigma = |\phi_{\epsilon}\rangle \langle \phi_{\epsilon}|$, $|\phi_{\epsilon}\rangle = \sqrt{\epsilon} |0\rangle + \sqrt{1-\epsilon} |1\rangle$, 
or, for $\gamma \geq 1$, consider
$\rho = |\phi_{\epsilon}\rangle \langle \phi_{\epsilon}|$ and $\sigma = \gamma^{-1} |0\rangle \langle 0 | + (1-\gamma^{-1}) |2\rangle \langle 2|$.
A direct calculation shows that
\begin{align}
    \frac{\Tr[\rho \wedge \gamma \sigma ]}{ \Tr[P\wedge \gamma Q] } 
    =\frac{1-\sqrt{1-\epsilon}}{\epsilon}
    \overset{\epsilon \searrow 0}{ \longrightarrow }
    \frac12.
\end{align}

On the other hand, the harmonic-mean lower bound $s\wedge (1-s) \cdot \Tr[P\mathbin{!}_s Q] \leq \Tr[A \wedge B]$ in Proposition~\ref{prop:NS_lower} is already sharp in the commuting case.
\end{remark}

\section{Representations for the Harmonic Means} \label{sec:representation}

\begin{proposition}[Equivalent representations] \label{prop:representations}
For positive semidefinite trace-class operators $A, B \ge 0$ and $s \in (0,1)$, the Petz--\NS harmonic mean $\HM_s(A,B)$ introduced in Definition~\ref{defn:harmonic_mean} has the following equivalent forms.
\begin{enumerate}[(i)]
    \item {Integral form:}
    Let $W$ be the solution of the Sylvester equation: $s A W + (1-s) W B = AB$ satisfying $W = A^0 W B^0$. Then,
    \begin{align}
    \HM_s(A,B)= \Tr[W] =
    \langle A, \left(s \mathscr{L}_A + (1-s) \mathscr{R}_B \right)^{-1}( B )\rangle
    =
    \int_{0}^{\infty} \Tr \left[  A \e^{-t s A} B \e^{-t (1-s) B} \right] \d t,
    \end{align}
    where $\mathscr{L}_A(X) \coloneqq AX$ and $\mathscr{R}_B(X) \coloneqq XB$ are the left and right multiplication, respectively.
    
    \item Supremum representation:
    \begin{align} \label{eq:supremum_representation}
    \HM_s(A,B) 
    = \sup_{Y = A^0 Y B^0} \left\{ \Tr[Y+Y^\dagger] - (1-s) \Tr[Y^\dagger A^{-1} Y] - s\Tr[ Y B^{-1}Y^{\dagger} ] \right\},
    \end{align}
    where $A^{-1}$ and $B^{-1}$ mean the pseudoinverses of $A$ and $B$, respectively.
    
    \item Infimum representation: 
    \begin{align} \label{eq:infimum_representation}
    \HM_s(A,B) = \inf_{X} \left\{ \frac{1}{1-s}\Tr[(\I-X)^{\dagger}A(\I-X)] + \frac{1}{s}\Tr[XBX^\dagger] \right\}.
    \end{align}
\end{enumerate}
\end{proposition}

\begin{remark}[Comparisons]
The Kubo--Ando harmonic operator mean \eqref{eq:KA_harmonic_mean} also admits the variational representations \cite[\S 1 \& 4]{Bha09}:
\begin{align}
    A \mathbin{!}_s B 
    &= \sup_{Y = (A^0 B^0)^{\infty} Y (A^0 B^0)^{\infty} } \left\{ Y+Y^\dagger - (1-s) Y^\dagger A^{-1} Y - s Y^\dagger B^{-1} Y \right\}
    \\
    &= \inf_{X} \left\{ \frac{1}{1-s}(\I-X)^{\dagger}A(\I-X) + \frac{1}{s}X^\dagger BX \right\},
\end{align}
where the optimizations are with respect to the L\"owner partial order, the pseudoinverses are employed, and $(A^0 B^0)^{\infty}$ means the projection onto the common support of $A$ and $B$ or the zero projector if there is none.
Notice that the differences compared to \eqref{eq:supremum_representation} and \eqref{eq:infimum_representation} are the support constraint and the conjugate transpose of the variable.
\end{remark}

\begin{remark}[BKM metric]
The Petz--\NS harmonic mean $\HM_s(A,B)$ acts as the fundamental resolvent slice of the Bogoliubov--Kubo--Mori (BKM) geometry \cite{Kubo1957, Mori1958, Bogoliubov1961}. 
Indeed, the integral is given by
\begin{align}
\int_0^1 \HM_s(A,B) \, \d s
&= \sum_{i,j} a_i b_j \left( \frac{\ln a_i - \ln b_j}{a_i - b_j} \right) |\langle u_i | v_j \rangle|^2
= \int_0^\infty \Tr\left[ A \left( A + t \I \right)^{-1} B \left( B + t \I \right)^{-1} \right]\d t,
\end{align}
which is the so-called relative BKM inner product $\int_0^\infty \Tr\big[ (\cdot)^{\dagger} \left( A + t \I \right)^{-1} (\cdot) \left( B + t \I \right)^{-1} \big]\d t$ between $A^{\nicefrac12} B^{\nicefrac12}$ and $A^{\nicefrac12} B^{\nicefrac12}$.
\end{remark}

\begin{proof}
We prove the three representations by working in the eigenbases of $A$ and $B$.
Write
\begin{align}
    A=\sum_i a_i |u_i\rangle\langle u_i|,\qquad
    B=\sum_j b_j |v_j\rangle\langle v_j|,
\end{align}
and set $c_{ij}\coloneqq \langle u_i|v_j\rangle$.
It suffices to consider the pairs $(i,j)$ with $a_i b_j>0$; the other summands vanish by the convention in Definition~\ref{defn:harmonic_mean}, and the pseudoinverses simply restrict all expressions to the supports of $A$ and $B$.

\begin{enumerate}[(i)]
    \item Let
    \[
        \mathcal{M}_s \coloneqq sL_A+(1-s)R_B,
        \qquad L_A(X)=AX,\quad R_B(X)=XB.
    \]
    Write
    \[
        A=\sum_i a_i |u_i\rangle\langle u_i|,
        \qquad
        B=\sum_j b_j |v_j\rangle\langle v_j|,
        \qquad
        c_{ij}\coloneqq \langle u_i|v_j\rangle .
    \]
    The rank-one operators $|u_i\rangle\langle v_j|$ diagonalize $\mathcal{M}_s$, since
    \[
        \mathcal{M}_s(|u_i\rangle\langle v_j|)
        =
        \left(sa_i+(1-s)b_j\right)|u_i\rangle\langle v_j|.
    \]
    Moreover,
    \[
        AB
        =
        \sum_{i,j} a_i b_j c_{ij}|u_i\rangle\langle v_j|.
    \]
    Hence the support-selected solution of
    \[
        sAW+(1-s)WB=AB,
        \qquad W=A^0WB^0,
    \]
    is
    \[
        W
        =
        \mathcal{M}_s^{-1}(AB)
        =
        \sum_{i,j}
        \frac{a_i b_j}{sa_i+(1-s)b_j}
        c_{ij}|u_i\rangle\langle v_j|.
    \]
    Taking the trace gives
    \[
        \Tr[W]
        =
        \sum_{i,j}
        \frac{a_i b_j}{sa_i+(1-s)b_j}
        |c_{ij}|^2
        =
        \HM_s(A,B).
    \]
    Since $L_A$ and $R_B$ commute,
    \[
        \e^{-t\mathcal{M}_s}(X)
        =
        \e^{-tsA}X\e^{-t(1-s)B}.
    \]
    Therefore
    \[
        W
        =
        \int_0^\infty \e^{-t\mathcal{M}_s}(AB)\,\d t
        =
        \int_0^\infty
        \e^{-tsA}AB\e^{-t(1-s)B}\,\d t.
    \]
    Taking the trace gives
    \[
        \HM_s(A,B)
        =
        \Tr[W]
        =
        \int_0^\infty
        \Tr\left[A\e^{-tsA}B\e^{-t(1-s)B}\right]\d t.
    \]

    \item We prove the supremum representation by the same diagonalization.  Let
    \[
        Y=\sum_{i,j} y_{ij}|u_i\rangle\langle v_j|,
        \qquad
        y_{ij}=0 \quad\text{whenever } a_i b_j=0.
    \]
    This is exactly the support constraint $Y=A^0YB^0$.
    Then
    \[
        \Tr[Y+Y^\dagger]
        =
        2\Re\sum_{i,j} y_{ij}\overline{c_{ij}},
    \]
    while
    \[
        \Tr[Y^\dagger A^{-1}Y]
        =
        \sum_{i,j}\frac{|y_{ij}|^2}{a_i},
        \qquad
        \Tr[YB^{-1}Y^\dagger]
        =
        \sum_{i,j}\frac{|y_{ij}|^2}{b_j}.
    \]
    Hence the objective function in the supremum formula becomes
    \[
        \sum_{i,j}
        \left[
        2\Re\left(y_{ij}\overline{c_{ij}}\right)
        -
        \left(
        \frac{1-s}{a_i}
        +
        \frac{s}{b_j}
        \right)
        |y_{ij}|^2
        \right].
    \]
    For each pair $(i,j)$, the scalar identity
    \[
        \sup_{y\in\mathbb{C}}
        \left\{
        2\Re(y\overline{c})
        -
        \alpha |y|^2
        \right\}
        =
        \frac{|c|^2}{\alpha},
        \qquad \alpha>0,
    \]
    gives, with
    \[
        \alpha
        =
        \frac{1-s}{a_i}+\frac{s}{b_j},
    \]
    the value
    \[
        \frac{|c_{ij}|^2}{\frac{1-s}{a_i}+\frac{s}{b_j}}
        =
        \frac{a_i b_j}{sa_i+(1-s)b_j}|c_{ij}|^2.
    \]
    Summing over $i,j$ gives $\HM_s(A,B)$.
    The maximizing coefficient is
    \[
        y_{ij}
        =
        \frac{c_{ij}}{\frac{1-s}{a_i}+\frac{s}{b_j}}
        =
        \frac{a_i b_j}{sa_i+(1-s)b_j}c_{ij},
    \]
    which is precisely the coefficient of the solution $W$ in part (i).

    \item Finally, write
    \[
        x_{ij}\coloneqq \langle u_i|X|v_j\rangle .
    \]
    Since
    \[
        \langle u_i|(\I-X)|v_j\rangle
        =
        c_{ij}-x_{ij},
    \]
    we have
    \[
        \Tr[(\I-X)^\dagger A(\I-X)]
        =
        \sum_{i,j}a_i|c_{ij}-x_{ij}|^2,
    \]
    and
    \[
        \Tr[XBX^\dagger]
        =
        \sum_{i,j}b_j|x_{ij}|^2.
    \]
    Therefore the objective function in the infimum formula becomes
    \[
        \sum_{i,j}
        \left[
        \frac{a_i}{1-s}|c_{ij}-x_{ij}|^2
        +
        \frac{b_j}{s}|x_{ij}|^2
        \right].
    \]
    For $a,b>0$ and $c\in\mathbb{C}$, completing the square gives
    \[
        \inf_{x\in\mathbb{C}}
        \left\{
        \frac{a}{1-s}|c-x|^2
        +
        \frac{b}{s}|x|^2
        \right\}
        =
        \frac{ab}{sa+(1-s)b}|c|^2.
    \]
    Applying this scalar identity with $a=a_i$, $b=b_j$, and $c=c_{ij}$ yields
    \[
        \inf_X
        \left\{
        \frac{1}{1-s}\Tr[(\I-X)^\dagger A(\I-X)]
        +
        \frac{1}{s}\Tr[XBX^\dagger]
        \right\}
        =
        \sum_{i,j}
        \frac{a_i b_j}{sa_i+(1-s)b_j}|c_{ij}|^2
        =
        \HM_s(A,B).
    \]
\end{enumerate}

The above argument is coefficient-wise on the support of $A$ and $B$. For positive semidefinite trace-class operators, the result follows by finite-rank approximations. This proves all three representations.
\end{proof}

\section{Yet Another Proof of The Quantum {Chernoff} Bound} \label{sec:Chernoff}

For any positive semidefinite trace-class operators $A, B \ge 0$, the trace inequality,
\begin{align} \label{eq:Chernoff}
    \Tr\left[A\wedge B\right]
    &\leq \Tr\left[ A^{1-s} B^s \right], \quad \forall\, s\in [0,1],
\end{align}
customarily called the \emph{quantum Chernoff bound} \cite{Che52, ACM+07, ANS+08}, is a foundational result of quantum hypothesis testing.
(The case $s = \nicefrac12$ refers to the \emph{Powers--St{\o}rmer inequality} \cite{powers1970free}.)
If $A\wedge B \geq 0$, the operator monotonicity of the power functions $(\cdot)^{1-s}$ and $(\cdot)^s$, $s\in [0,1]$, would directly imply the bound~\eqref{eq:Chernoff}, i.e.,
\begin{align} \label{eq:Chernoff_attempt}
    \Tr\left[A\wedge B\right]
    &= \Tr\left[ (A\wedge B)^{1-s} (A\wedge B)^s \right]
    \leq \Tr\left[ A^{1-s} B^s \right],
\end{align}
because $A\wedge B \leq A$ and $A\wedge B \leq B$ by definition.
However, $A\wedge B \geq 0$ is not always the case due to noncommutativity, hindering the attempt in \eqref{eq:Chernoff_attempt}.
This makes the first proof of the quantum Chernoff bound highly nontrivial and substantial 
\cite[Theorem 1]{ACM+07}, \cite[Theorem 2]{ANS+08}.

A surprisingly elegant and simple proof based only on the operator monotonicity and the cyclic property of trace is given by Narutaka Ozawa \cite[Proposition 1.1]{JOP+12} (see also \cite[p.~270]{JOP+12_book} and \cite[Theorem 4.54]{HP14}).
Later, Li's multiple Chernoff bound \cite[Theorem 2]{Li16} implies \eqref{eq:Chernoff} (with an additional factor depending on the number of distinct eigenvalues).
Recently, \eqref{eq:Chernoff} can be proved via the integral representation of the quantum hockey-stick divergence \cite[Lemma 3.11]{beigi2025some} and the layer cake representation
\cite[Theorem 5.4]{LHC25_layer_cake}.
Our harmonic-mean bound established in Theorem~\ref{theorem:harmonic-mean_bound} already implies \eqref{eq:Chernoff}.
Inspired by the analysis, here we provide yet another proof of the quantum Chernoff bound by employing a supremum representation of $\Tr[A^{1-s}B^s]$ given below, without relying on the intermediate harmonic mean.

\begin{proposition} \label{prop:Chernoff_another}
		For any positive semidefinite trace-class operators $A, B \ge 0$ and $s \in [0,1]$, 
		\begin{align}
			\Tr\left[ A^{1-s} B^s \right] 
            &= \sup_{Y = A^0 Y B^0} \left\{ \Tr[Y + Y^\dagger] - \Tr\left[Y^\dagger A^{-(1-s)} Y B^{-s}\right] \right\},
		\end{align}
		where the negative powers are with respect to the pseudoinverse, and the supremum is attained at $Y = A^{1-s} B^s$.
\end{proposition}
\noindent We postpone the proof to the end of this section.

Using the representation given in Proposition~\ref{prop:Chernoff_another}, we choose a non-Hermitian feasible solution $Y_0 = A (A\vee B)^{-1} B$.
Recall Lemma~\ref{lemma:min_identity}, we denote by
\begin{align}
    \Min 
    \coloneqq \Tr\left[ A\wedge B\right]
    = \Tr\left[Y_0\right] = \Tr\big[Y_0^\dagger\big].  
\end{align}
Then, the representation given in Proposition~\ref{prop:Chernoff_another} implies that 
\begin{align}
    \Tr\left[ A^{1-s} B^s \right] 
    &\geq \Tr[Y_0 + Y_0^\dagger] - \Tr\left[Y_0^\dagger A^{-(1-s)} Y_0 B^{-s}\right]
    \\
    &= 2 \Min - \Tr\left[ B (A\vee B)^{-1} A^{1+s} (A\vee B)^{-1} B^{1-s}\right].
\end{align}
To prove the quantum Chernoff bound \eqref{eq:Chernoff}, it remains to show
\begin{align} \label{eq:Chernoff_intermediate}
    \Tr\left[ B (A\vee B)^{-1} A^{1+s} (A\vee B)^{-1} B^{1-s}\right]
    \leq \Min, \quad \forall\, s\in[0,1].
\end{align}

To this end, let 
\begin{align}
    f(s)
    &\coloneqq 
    \Tr\left[ B (A\vee B)^{-1} A^{1+s} (A\vee B)^{-1} B^{1-s}\right]
    \\
    &= \sum_{i,j} \left( a_i b_j^2 \right) \left( \frac{a_i}{b_j} \right)^s |\langle u_i | (A\vee B)^{-1} | v_j \rangle|^2,
\end{align}
which is convex in $s \in [0,1]$.
Hence, we only need to upper bound the endpoints $f(0)$ and $f(1)$, i.e.,
\begin{align}
    f(s) 
    \leq f(0) \vee f(1) \leq \Min, \quad \forall\, s\in[0,1].
\end{align}
%
Let $H = A-B$, $H_+ = (A-B)_+$, and $H_- = (B-A)_+$.
Using $A \vee B - B = H_+$, $A = A\vee B - H_-$, and $H_- H_+ = 0$, we calculate
\begin{align}
\Min - f(0)
&=\Tr\left[ A (A\vee B)^{-1}B \right]-\Tr\left[A (A\vee B)^{-1} B^2 (A\vee B)^{-1} \right]
\\
&=\Tr\left[A(A\vee B)^{-1}B\left(A\vee B-B\right)(A\vee B)^{-1}\right]
\\
&=\Tr\left[(A\vee B-H_-)(A\vee B)^{-1}(A\vee B-H_+) H_+ (A\vee B)^{-1}\right] 
\\
&=\Tr\left[(A\vee B-H_-)(H_+ - (A\vee B)^{-1}H_+^2) (A\vee B)^{-1}\right] 
\\
&=\Tr\left[ (A\vee B)H_+ (A\vee B)^{-1} - H_+^2 (A\vee B)^{-1} + H_- (A\vee B)^{-1} H_+^2 (A\vee B)^{-1} \right] 
\\
&= \Tr\left[H_+ - H_+ (A\vee B)^{-1} H_+ \right] + \Tr\left[ H_- (A\vee B)^{-1} H_+^2 (A\vee B)^{-1} \right].
\end{align}
The first term $\Tr\left[H_+-H_+(A\vee B)^{-1}H_+\right]$ is non-negative by \eqref{eq:H_+_monotone}.
We hence prove $f(0) \leq \Min$.
The case $f(1) \leq \Min$ follows in a similar fashion by switching the roles of $A$ and $B$ and noting that $A\vee B = B\vee A$.
This proves \eqref{eq:Chernoff_intermediate} and hence 
completes the proof of the quantum Chernoff bound \eqref{eq:Chernoff}.

\begin{proof}[Proof of Proposition~\ref{prop:Chernoff_another}]
For every operator $Y$ satisfying $Y = A^0 Y B^0$ and any $s\in[0,1]$,
\begin{align}
\Tr\left[ \left( A^{\nicefrac{-(1-s)}{2}} Y B^{\nicefrac{-s}{2}} - A^{\nicefrac{(1-s)}{2}} B^{\nicefrac{s}{2}} \right)^\dagger \left( A^{\nicefrac{-(1-s)}{2}} Y B^{\nicefrac{-s}{2}} - A^{\nicefrac{(1-s)}{2}} B^{\nicefrac{s}{2}} \right) \right] \geq 0.
\end{align}
Expanding the square and rearranging the terms explicitly gives the ``$\leq$'' bound of the supremum formula for every feasible $Y$. Furthermore, equality is attained at $Y = A^{1-s}B^s$, which gives the reverse inequality and hence establishes the supremum.
\end{proof}

{\larger
\bibliographystyle{myIEEEtran}
\bibliography{reference}
}

\end{document}